\documentclass{elsart}               

\usepackage{graphicx}          
\usepackage[dvips]{epsfig}    
\usepackage{amsmath}
\usepackage{amssymb}
\usepackage{mathtools}
\usepackage{epsfig}
\usepackage{epstopdf}
\usepackage{hyperref}
\usepackage{algorithm}
\usepackage{algorithmicx}
\usepackage{algpseudocode}

\def \0{\mathbf{0}}
\def \1{\mathbf{1}}

\def \I{\mathbf{I}}
\def \K{\mathbf{K}}

\def \bK{\bar{\mathbf{K}}}
\def \hK{\hat{\mathbf{K}}}
\def \hG{\hat{G}}
\def \bG{\bar{G}}
\def \hF{\hat{F}}

\def \bP{\mathbf{P}}

\def \cP{\mathcal{P}}

\def \cA{\mathcal{A}}

\def \cU{\mathcal{U}}

\def \bR{\mathbb{R}}

\def \E{\textrm{E}}

\def \Var{\textrm{Var}}
\def \LQ{\textrm{LQ}}
\def \MV{\textrm{MV}}
\def \cP{\mathcal{P}}
\def \cU{\mathcal{U}}
\def \cA{\mathcal{A}}
\def \F{\mathcal{F}}
\def \L{\mathcal{L}}
\def \K{\mathcal{K}}
\def \bR{\mathbb{R}}
\def \bP{\mathbb{P}}
\def \E{\textrm{E}}
\def \Var{\textrm{Var}}
\def \1{\mathbf{1}}
\def \bK{\bar{K}}
\def \hK{\hat{K}}
\def \hG{\hat{G}}
\def \bG{\bar{G}}
\def \L{\mathcal{L}}
\def \LQ{\textrm{LQ}}
\def \MV{\textrm{MV}}

\def \tK{\tilde{K}}

\def \0{\mathbf{0}} 
\def \I{\mathbf{I}} 

\newtheorem{lemma}{Lemma}

\newtheorem{theorem}{Theorem}

\newtheorem{assumption}{Assumption}
\newtheorem{example}{Example}

\begin{document}

\begin{frontmatter}

\title{Optimal Control of Constrained Stochastic Linear-Quadratic Model with Applications\thanksref{footnoteinfo}} 

\thanks[footnoteinfo]{This research work was partially supported by National Natural Science Foundation of China under the grant 61573244.}

\author[SJTU]{Weiping Wu}\ead{godream@sjtu.edu.cn},                        
\author[SHUFE]{Jianjun Gao\corauthref{cor1}}                              
\corauth[cor1]{Corresponding Author.}
\ead{gao.jianjun@shufe.edu.cn},
\author[SJTU]{Junguo Lu}\ead{jglu@sjtu.edu.cn},                   
\author[POLYU]{Xun Li}\ead{li.xun@polyu.edu.hk}

\address[SJTU]{Department of Automation, Shanghai Jiao Tong University, Shanghai, China}  
\address[SHUFE]{School of Information Management and Engineering, Shanghai University of Finance and Economics, Shanghai, China}             
\address[POLYU]{Department of Applied Mathematics, The Hong Kong Polytechnic University, Hong Kong, China}

\begin{abstract}
This paper studies a class of continuous-time scalar-state stochastic Linear-Quadratic (LQ) optimal control problem with the linear control constraints. Applying the state separation theorem induced from its special structure, we develop the explicit solution for this class of problem. The revealed optimal control policy is a piece-wise affine function of system state. This control policy can be computed efficiently by solving two Riccati equations off-line. Under some mild conditions, the stationary optimal control policy can be also derived for this class of problem with infinite horizon. This result can be used to solve the constrained dynamic mean-variance portfolio selection problem. Examples shed light on the solution procedure of implementing our method.

\end{abstract}

\end{frontmatter}

\section{Introduction}
In this paper, we study the constrained Linear-Quadratic(LQ) control problem for the continuous-time stochastic scalar-state system. Without a doubt, the LQ type of control model (see, e.g., \cite{Kalman:1960}\cite{Wonham:1969}\cite{Bismut:1976}) plays an important role in both theoretical study and various applications(\cite{AndersonLQ:2014}\cite{Pham:2009}\cite{ZhouLi:2000}\cite{LiZhouLim:2002}), largely due to its elegant structure. From the solution point of view, the continuous-time LQ optimal control is achieved by solving of an unconstrained convex quadratic optimization problem,\footnote{For the classical discrete-time LQ control model, we need to solve a sequence of unconstrained convex quadratic programming problems.} which admits the closed form solution. The linear feedback control policy and the correspondent Riccati Equation arise naturally from this structure. However, in real applications, the control constraints are inevitable, e.g., the system usually subjects to the physical limits of the executor or the economic regulation restrictions. Generally speaking, involving the control constraints usually destroy the elegant structure of the LQ control, i.e., the closed form of control policy is not available any more.

In the literature, there are many efforts on investigating the control policy for the constrained stochastic LQ control problem. However, even for the deterministic system with a simple nonnegative control constraints(e.g., the control policy is nonnegative), there is no closed form solution(see, e.g., \cite{Campbell:1982}\cite{HeemelsEijndhovenStoorvogel:1998}). As for the discrete-time constrained LQ control model, the analytical control police can be developed only for some special cases. For example, Gao $et~al.$\cite{GaoLi:2011} study the deterministic LQ control model with cardinality constraints on the total control periods. They apply the dynamic programming to develop a semi-analytical control policy. Bemporad $et~al.$\cite{Bemporad:2002} investigate the deterministic LQ model with a general linear inequality constraints on both states and controls. By utilizing the proposed parametric programming method, they derive an explicit solution. However, this approach may suffer a heavy computational burden even when the size of the problem is not large. Recently, Wu et~al. \cite{WuGaoLiShi:2017} adopt the state separation theorem and provide the explicit control policy for a class of constrained discrete-time stochastic LQ optimal control problem for the system with multiplicative noise. Except these special cases, the method so called Model Predictive Control(MPC)(e.g., see \cite{BernardiniBemporad:2012}\cite{Patrinos:2014}) is well established to characterize the approximated control policy. For the recent development of the MPC method on solving the constrained stochastic LQ control model, the reader may refer \cite{Primb:2009} \cite{BernardiniBemporad:2012} \cite{Mesbah:2016} for a comprehensive survey.

To our best knowledge, in current literature, there are only few works focusing on investigating the analytical control policy for the continuous-time constrained stochastic LQ control problem. Although the optimality conditions can be established for this class of problem by using the general optimal control theory(see, e.g., \cite{YongZhou:1999}\cite{Pham:2009}), the analytical control policy is only available for some special structured problems. The dynamic Mean-Variance(MV) portfolio selection problem is one of these cases. The seminar work \cite{ZhouLi:2000} first extends the classical MV portfolio selection model to the continuous-time setting and develop the optimal portfolio policy by using the stochastic control approach. Following this line, Li et al.\cite{LiZhouLim:2002} solve the portfolio optimization problem with the no-shorting constraint(or equivalently, the nonnegative control constraint). They reformulate such a problem into a constrained LQ control problem and develop the semi-analytical solution by constructing the viscosity solution for the associated Hamilton-Jacobi-Bellman (HJB) equation. Based on dynamic programming, Cui et al. \cite{CuiGaoLiLi:2014} \cite{CuiLiLi:2017} developed solution of this type of problem under the discrete-time model with no shorting constraints and cone constraints, respectively. As for the constrained stochastic LQ control problem, Hu et al. \cite{HuZhou:2005} adopt the backward stochastic differential equation (BSDE) approach to characterize the analytical solution for the problem with cone constraint on control. Recently, Hu et al.\cite{HuHuangNie:2017} characterize the analytical solution of the LQ mean-field game with input constraints.

Our work contributes to current literature on developing the control policy of the constrained stochastic LQ control problem. The contribution of this work includes several aspects. Firstly, for the problem with finite time horizon, we successfully derive the analytical solution for a class of stochastic LQ control problems with a general linear constraints. This type of constraint goes beyond the cone constraints considered in \cite{CuiLiLi:2017}\cite{HuZhou:2005} and the nonnegative constraints considered in \cite{CuiGaoLiLi:2014}\cite{LiZhouLim:2002}. We show that the optimal control policy is a piece-wise affine function of state variable, which can computed efficiently by solving two Riccati equations off-line. Secondly, for the problem with infinite time horizon, we show that, under some mild condions, the stationary optimal control policy can be obtained. When the stationary control exists, we also provide an algorithm to identify the solutions the associated algebraic Riccati equations. Last but not least, we apply such a optimal control model to solve the constrained dynamic MV portfolio selection problem. Based on the real market data, we also provide the example to illustrate how to use our solution method in real application.


This paper is organized as follows. Section \ref{se_formulation} provides the basic formulation the model of continuous-time constrained stochastic LQ control problem for both finite and infinite time horizon. The analytical solutions of these two problems are presented in Section \ref{se_sol_P} and Section \ref{se_sol_Pinf}, respectively. Section \ref{se_application_MV} devotes to deriving the optimal investment policy and efficient frontier for the constrained dynamic MV portfolio selection problem. Section \ref{se_example} provides some examples to show that how to use our solution scheme in real applications. Finally, Section \ref{se_conclusion} concludes the whole paper and provides some further extensions.

\textit{Notation}. In this paper, we use the following basic notations: $\bR^n$ denotes the set of $n$-dimensional real column vectors, $\bR^{n \times m}$ denotes the set of $n \times m$ real matrices, $R\succeq 0$$(R \succ 0)$ denotes a positive semidefinite (positive definite) matrix; $\0$ and $\I$ denote zero matrix(vector) and identify matrix in a proper dimension, respectively; $\1_{\mathcal{A}}$ denotes the indicator function such that $\1_{\mathcal{A}}=1$ if the condition $\mathcal{A}$ holds true and $\1_{\mathcal{A}}=0$ otherwise. We use $v(\cP)$ to denote the optimal objective value for any problem $(\cP)$.

\section{Problem Formulation} \label{se_formulation}

\subsection{Problem Formulation with Finite Time Horizon}\label{sse_finite}
In this work, we consider the following scalar-state linear stochastic differential
equation (SDE), for $t\in[0,T]$,
\begin{align}
\begin{cases}
dx(t)=\big( A(t)x(t)+B(t)^{\prime}u(t) \big)dt\\
~~~~~~~~~~+\big( x(t)C(t)^{\prime}+u(t)^{\prime}D(t) \big) dW(t),\\
x(0)=x_0,
\end{cases}\label{def_xt}
\end{align}
where $T$ is a finite number, $A(t) \in \bR$, $B(t) \in \bR^n$, $C(t) \in \bR^m$, $D(t) \in \bR^{n \times m}$ are deterministic and bounded system matrices, $u(t) \in \bR^n$ is the control vector, $x(t)\in \bR$ is the system state and $W(\cdot)$ is the $m$-dimensional standard Brownian motion. All the randomness are modeled by a complete filtrated probability space $(\Omega,\F, \{\F_t\}_{t\geq0},\bP)$, where $\F_t$ is the augmented $\sigma$-algebra generated by $W(t)$, which represents the information set available at time $t$, $t\in [0,T]$. We define $\L_{\F}^{2}(0,T;\bR^k)$ as the set of $\F_t$ adapted square integrable stochastic process, i.e., set of $\{\F_{t}\}_{t\geq0}$-adapted stochastic process $f(t)$ satisfying $\E[\int_{0}^{T}\|f(t)\|^2dt]$$<$$+\infty$. In the following part, we use $\E_t[\cdot]$ to denote the conditional expectation $\E[\cdot |\F_t]$.

We assume admissible controls $u(\cdot)$ are square integrable, i.e., $u(\cdot)$$\in$$L_{\F}^2(0,T;\bR^n)$. In addition, motivated from the real applications, the admissible controls satisfy the following control constraint,
\begin{align}
\cU_t(x(t))=&\Big\{~u(t)\in \L_{\F}^{2}(0,T;\bR^n)\big|~H(t)u(t)\leq d(t)|x(t)|~\Big\} \label{def_Ut}
\end{align}
for $t \in [0,T]$, where $H(t) \in \bR^{k\times n}$ and $d(t)\in \bR^k$ are given deterministic matrices and vectors, respectively. \footnote{Due to states $x(t)$ and controls $u(t)$ are random variables for all $t\in(0,T]$, we assume that the inequalities given in (\ref{def_Ut}) should be held \textit{almost surely}, i.e., the inequalities are invalid for the cases with zero probability measure. In order to simplify the notation, we omit the term `almost surely' in this paper.} To ensure the existence of the feasible control policies, we need the following assumption.
\begin{assumption}\label{asmp_feasibility_finite}
The following set,
\begin{align*}
\Big\{K \in \mathbb{R}^n~|~H(t) K \leq d(t), H(t) K\leq \0 \Big\},
\end{align*}
is nonempty for all $t\in [0,T]$.
\end{assumption}
It is not hard to see that, under the Assumption \ref{asmp_feasibility_finite}, the feasible set $\cU_t(x(t))$ is always nonempty. In this model, the constraint set ($\ref{def_Ut}$) has the freedom to model several types of control constraints. If we set $d(t)=\0$ and $H(t)=-\I$, the constraint (\ref{def_Ut}) becomes the nonnegative constraint, i.e., $u(t)\geq \0$. If we set $H(t)$$=$$\left(
            \begin{array}{c}
              \I \\
              -\I \\
            \end{array}
          \right)$
and $d(t)$$=$$\left(
    \begin{array}{c}
        d_1(t) \\
        -d_2(t) \\
        \end{array}
        \right)$
with $d_1(t)$$>$$d_2(t)$ for all $t\in[0,T]$ in constraint (\ref{def_Ut}), it is equivalent to the lower and upper bound constraints,
\begin{align*}
d_2(t)|x(t)| \leq u(t) \leq d_1(t)|x(t)|,
\end{align*}
for all $t\in[0,T]$. Furthermore, if we set $d(t)=\0$, the constraint (\ref{def_Ut}) becomes the cone constraint.


In order to formulate the performance index, we introduce the following deterministic matrices and vectors, $R(t)$ $\in$ $\bR^{n\times n}$ with $R(t)$ $\succeq$ $0$, $q(t)$ $\in$ $\bR$ with $q(t)$$\geq$$0$, and $S(t)$ $\in$ $\bR^n$ for $t\in[0,T)$, which represent the penalties on controls, system states and the cross term, respectively. Let $q_T\in \bR$ with $q_T\geq 0$ be the penalty of the terminal state. These parameters can be written in a more compact formulation,
$Q(t)$ $:=$ $\left(\begin{array}{cc}
R(t)          & S(t)\\
S(t)^{\prime} & q(t)
\end{array}\right)$ for $t\in[0,T)$. We propose the inequality constrained stochastic LQ control problem (ICLQ) in continuous time as follows,
\begin{align*}
(\cP^T_{\LQ})&~\min_{\{u(t)\}|_{t=0}^{T}}~\E \left[ \int_{0}^{T}
\left(\begin{array}{cc}
u(t)\\
x(t)
\end{array}\right)^{\prime}Q(t)
\left(\begin{array}{cc}
u(t)\\
x(t)
\end{array}\right)dt~\right] +\E[q_Tx(T)^2]\\
{\rm s.t.}~&\{x(t),u(t)\}~\textrm{satisfies (\ref{def_xt}) and (\ref{def_Ut})}.\notag
\end{align*}

To solve the above problem, we need the following assumption.
\begin{assumption} \label{asmp_psd}
Assume that $Q(t)\succeq 0$ and $D(t)D(t)^{\prime}\succ 0$ for $t\in[0,T]$.
\end{assumption}
Note that the Assumption \ref{asmp_psd} requires matrix $D(t)D(t)^{\prime}$ to be nonsingular. This type of assumption is also imposed in the continuous-time MV portfolio selection problem which requires the volatility coefficient $\sigma(t)$ to be nondegenerate(please refer \cite{LiZhouLim:2002}\cite{ZhouLi:2000} \cite{HuZhou:2005}). Furthermore, the Assumption \ref{asmp_psd} also guarantees the convexity of problem ($\cP_{\LQ}^T$).

\subsection{Problem Formulation with Infinite Time Horizon}\label{sse_infinite}
We are also interested in the stationary control for the problem $(\cP_{\LQ}^T)$ when the horizon, $T$ goes to infinite. In this case, we assume all system parameters, $A(t)$, $B(t)$, $C(t)$ and $D(t)$ for $t\geq 0$ to be time invariant, i.e., $A(t)=A$, $B(t)=B$, $C(t)=C$ and $D(t)=D$ for $t \geq 0$, where $A$, $B$, $C$ and $D$ are constant matrices(vectors) in a proper dimension. The system dynamic (\ref{def_xt}) becomes
\begin{align}
\begin{cases}
dx(t)=\big( Ax(t)+B^{\prime}u(t) \big)dt+\big( x(t)C^{\prime}+u(t)^{\prime}D \big) dW(t),\\
x(0)=x_0
\end{cases} \label{def_xt_inf}
\end{align}
for $t\geq 0$. As for the constraint (\ref{def_Ut}), we also assume all the parameters to be constant, i.e.,
\begin{align}
\tilde{\cU}_t(x(t))&=\Big\{u(t)\in \L_{\F}^{2}(0,\infty;\bR^n)~|~H u(t)\leq d|x(t)|~\Big\}, \label{def_Ut_inf}
\end{align}
for $t\geq 0$, where $H\in \bR^{k \times n}$ and $d \in \bR^{n}$ are constant matrix and vector, respectively. Similar to the Assumption \ref{asmp_feasibility_finite}, we propose the following assumption to ensure the feasibility of the constraint set,
\begin{assumption}\label{asmp_feasibility_infinite}
The following set,
\begin{align*}
\Big\{K \in \mathbb{R}^n~|~H K \leq d, ~H K\leq \0 \Big\},
\end{align*}
is nonempty.
\end{assumption}
As for the cost function, we assume that all matrices to time invariant, i.e., $Q(t)$ $=$ $Q$ $=$
$\left(\begin{array}{cc}
R          & S\\
S^{\prime} & q
\end{array}\right)$  for any $t\in[0,\infty)$, where $R\in\bR^{n\times n}$, $S\in\bR^{n}$ and $q\in\bR $ are constant matrices, vectors and scalar value, respectively. We then consider the following problem with infinite time horizon,
\begin{align*}
(\cP_{\LQ}^{\infty})~&~\min_{\{u(t)\}|_{t=0}^{\infty}}~\E \left[ \int_{0}^{\infty}
\left(\begin{array}{cc}
u(t)\\
x(t)
\end{array}\right)^{\prime}Q
\left(\begin{array}{cc}
u(t)\\
x(t)
\end{array}\right)dt~\right] \\
{\rm s.t.}~&~\{x(t),u(t)\}~\textrm{satisfies (\ref{def_xt_inf}) and (\ref{def_Ut_inf})}.\notag
\end{align*}
For problem $(\cP_{\LQ}^{\infty})$, we need a stronger assumption by requiring $Q$ to be positive definite as,
\begin{assumption} \label{asmp_C}
Assume that $Q\succ 0$ and $DD^{\prime}\succ 0$.
\end{assumption}

\section{Solution Procedure for Problem (\texorpdfstring{$\cP_{\LQ}^T$}{Lg})}\label{se_sol_P}
In this section, we firstly provide an important theorem, which characterizes the state separation property for our formulation. Then, we apply this result to solve problem ($\cP_{\LQ}^T$).

\subsection{State Separation Property}\label{sse_state_separation}
To develop the analytical solution of problem ($\cP_{\LQ}^T$), we need the following state separation theorem, which can be regarded as a generalization of the ones developed in \cite{LiZhouLim:2002}. We  introduce the following set associated with the control constraint $\cU_t(x(t))$ as,
\begin{align*}
\K(t)=\big\{K \in \bR^n ~|~ H(t)K \leq d(t)\big\},
\end{align*}
for $t\in [0,T]$.
\begin{theorem}\label{thm_sep}
At any time $t\in [0,T]$, given a matrix $\Omega \in \bR^{n\times n}$ with $\Omega \succeq 0$ and a vector $\omega \in \mathbb{R}^n$, suppose the solutions and optimal objective values of the following problems,
\begin{align*}
&(\hat{\cP})~\min_{K\in \K(t)}~K^{\prime}\Omega K + 2\omega^{\prime}K ~~\textrm{and}~~\\
&(\bar{\cP})~\min_{K\in \K(t)}~K^{\prime}\Omega K - 2\omega^{\prime}K,
\end{align*}
are $\hat{K}$ with $v(\hat{\cP})$ and $\bar{K}$ with $v(\bar{\cP})$, respectively. Then the solution of the following problem for some $\alpha \in \mathbb{R}$,
\begin{align*}
(\cP(\alpha))~~\min_{u\in \cU_t(\alpha)}~u^{\prime}\Omega u + 2\alpha \omega^{\prime}u
\end{align*}
is
\begin{align*}
u^*(\alpha)= \big( \alpha \1_{\{ \alpha\geq 0 \}}\big)\hat{K}+ \big(|\alpha| \1_{\{\alpha<0\} }\big)\bar{K}
\end{align*}
with the optimal objective value
\begin{align}
v(\cP(\alpha))=\alpha^2 \big(v(\hat{\cP}) \1_{\{\alpha\geq 0\}}+v(\bar{\cP}) \1_{\{\alpha <0\}} \big). \label{thm_sep_obj}
\end{align}
\end{theorem}
\begin{proof}
We first consider the case of $\alpha\geq 0$. Introducing Lagrange multipliers $\beta \geq \0$ for the constraint $H(t) u $ $\leq$ $|\alpha| d(t)$ in problem ($\cP(\alpha)$) the Karush-Kuhn-Tucker(KKT) condition(e.g., see \cite{Boyd:2004}) of this problem is,
\begin{align}
\begin{dcases}
u^*= \Omega^{-1}(\alpha \omega + H(t)^{\prime} \beta), \\
H(t)u^* \leq d(t) \alpha, \\
\big( H(t)u^*-\alpha d(t)\big)_i  \beta_i =0,~~\beta_i\geq 0, ~~\textrm{for}~~i=1,\cdots,k,
\end{dcases}\label{KKT1}
\end{align}
 where $\beta$$=$$(\beta_1, \cdots, \beta_k)^{\prime}$, ${u}^*$$=$$(u^*_1,\cdots, u^*_n)^{\prime}$ and $\big(H(t){u}^* - \alpha d(t)\big)_i$ is the $i$-th element of the vector $H(t){u}^* - \alpha d(t)$ for all $i=1,\cdots,k$. Similarly, we introduce the lagrange multiplier $\rho$ $=$ $(\rho_1,\cdots, \rho_n)^{\prime}$ $\geq$ $\0$ for problem ($\hat{\cP}$). The correspondent KKT condition is as follows,
\begin{align}
\begin{dcases}
\hat{K}= \Omega^{-1}(\omega+H(t)^{\prime}{\rho}), \\
 H(t) \hat{K} \leq d(t), \\
\big( H(t)\hat{K} - d(t) \big)_i {\rho}_i =0,~~{\rho}_i\geq0,~~\textrm{for}~~i=1,\cdots,k,
\end{dcases}\label{KKT2}
\end{align}
where $\rho$$=$$(\rho_1, \cdots, \rho_k)^{\prime}$, $\hat{K}$$=$$(\hat{K}_1,\cdots, \hat{K}_n)^{\prime}$ and $(H(t)\hat{K}-d(t))_i$ is the $i$-th element of $H(t)\hat{K}-d(t)$ for all $i=1,\cdots,k$. Comparing the two KKT conditions (\ref{KKT1}) and (\ref{KKT2}), we are able to construct the solution of the KKT condition (\ref{KKT1}) by using the solution of (\ref{KKT2}). More specifically, suppose $\rho^* $ and $\hat{K}^*$ solve the KKT condition (\ref{KKT2}). If we let ${\beta}^*$ $=$ $\alpha \rho^* $ and $u^*$ $=$ $\alpha \hat{K}$, then it can easily verified that pair $\{{u}^*, {\beta}^*\}$ satisfies the KKT condition (\ref{KKT1}). That is to say, $u^*$ $=$ $\alpha \hat{K}$ solves the problem ($\cP(\alpha)$) with the optimal value being $g(u^*,\alpha) =\alpha^2 \hat{g}(\hat{K})$. For the case of $\alpha<0$, we may adopt the similar analysis to prove the correspondent result. Thus, we omit the detail.
\end{proof}

Theorem \ref{thm_sep} shows that we can construct the solution of problem ($\cP(\alpha)$) by using the solutions of problems ($\hat{\cP}$) and ($\bar{\cP}$). This result turns out to be important for the stochastic dynamic optimization problem. For example, if $\alpha \in \bR$ represents the system state (allowing for positive or negative) in a stochastic system, solving problem ($\cP(\alpha)$) directly is difficult, since $\alpha$ is a random variable. However, by using the Theorem \ref{thm_sep}, we could solve two deterministic optimization problems $(\hat{\cP})$ and $(\bar{\cP})$ off-line for $\hat{K}$ and $\bar{K}$, respectively, and then construct the solution of the problem  ($\cP(\alpha)$).


\subsection{Solution of problem (\texorpdfstring{$\cP^{T}_{LQ}$}{Lg})}
As for the classical unconstrained stochastic LQ control problem( e.g., see \cite{YongZhou:1999}), it is well known that the solution of the correspondent Hamilton-Jacobi-Bellman (HJB) equation takes the quadratic form with respect to system state, which enables us to identify the analytical optimal feedback control law. However, due to the constraints in our model ($\cP_{\LQ}^T$), the traditional method can no longer be applied. Furthermore, the correspondent HJB equation does not admit smooth solution any more. To conquer this difficulty, we adopt the concept of \textit{viscosity solution} of the PDE for the constrained HJB equation induced from our model $(\cP_{\LQ}^T)$.

We first define the value function of problem ($\cP_{\LQ}^T$) at any time $t\in [0,T]$ as follows,
\begin{align*}
V(t,x)&:=\min_{\{u(s)\}|_{s=t}^{T}}~\E_t\left[ \int_{t}^{T}
\left(\begin{array}{cc}
u(s)\\
x(s)
\end{array}\right)^{\prime} Q(s)
\left(\begin{array}{cc}
u(s)\\
x(s)
\end{array}\right)ds~\right]+q_T x(T)^2\\
s.t.~&\begin{cases}
dx(s)=\big( A(s)x(s)+B(s)^{\prime}u(s) \big) ds\\
~~~~~~~~~~+\big( x(s)C(s)^{\prime}+u(s)^{\prime}D(s) \big)dW(s),\\
x(t)=x,
\end{cases} \\
&u(s)\in \cU_s(x(s)),~~s\in [t,T],
\end{align*}
From the classical optimal control theory( \cite{YongZhou:1999}, \cite{Pham:2009}), the value function  $V(t,x)$ satisfies the following HJB equation,\footnote{ The notation $V_x(t,x)$, $V_t(t,x)$ and $V_{xx}(t,x)$ denote the first and second order partial derivatives of the value function $V(t,x)$.}
\begin{align}
&V_t(t,x)+\inf_{u(t) \in \cU_t(x(t))} \Big\{\frac{1}{2}V_{xx}(t,x)\Big(C(t)^{\prime}C(t)x^2\notag\\
&~~~+2xC(t)^{\prime}D(t)^{\prime}u + u^{\prime}D(t)D(t)^{\prime}u\Big) + V_x(t,x)\notag\\
&~~~\times \Big(A(t)x+B(t)^{\prime}u\Big)+\left(\begin{array}{cc}
u\\
x
\end{array}\right)^{\prime}Q(t)\left(\begin{array}{cc}
u\\
x
\end{array}\right) \Big\} =0\label{def_HJB}
\end{align}
with the boundary condition $V(T,x)=q_Tx^2$. Before we give the main results, we first introduce the following two ordinary differential equations(ODEs) for two unknown functions $\hG(t)$ and $\bG(t)$
as follows,\footnote{ The notation $\dot{G}(t)$ represents $\frac{d G(t)}{dt}$ for some differentiable function $G(t)$.}
\begin{align}
\dot{\hat{G}}(t)=& -\hat{G}(t)C(t)^{\prime}C(t) - 2\hat{G}(t)A(t) - q(t) - \min_{K(t) \in \K(t)} \hat{f}(t,K(t),\hat{G}(t)) \label{def_hatGt}\\
\dot{\bar{G}}(t)=& -\bar{G}(t)C(t)^{\prime}C(t) - 2\bar{G}(t)A(t) - q(t) - \min_{K(t) \in \K(t)} \bar{f}(t,K(t),\bar{G}(t)) \label{def_barGt}
\end{align}
with the terminal conditions $\hat{G}(T)=q_T$ and $\bar{G}(T)=q_T$, where $\hat{f}(t,K(t),\hat{G}(t))$ and $\bar{f}(t,K(t),\bar{G}(t))$ are defined as follows, respectively,
\begin{align*}
\hat{f}(t,K(t),&\hat{G}(t))=K(t)^{\prime}[\hat{G}(t)D(t)D(t)^{\prime}+R(t)]K(t)\\
 & + 2\big(\hat{G}(t)C(t)^{\prime}D(t)^{\prime}+\hat{G}(t) B(t)^{\prime}+S(t)^{\prime} \big) K(t), \\
\bar{f}(t,K(t),&\bar{G}(t))=K(t)^{\prime}[\bar{G}(t)D(t)D(t)^{\prime}+R(t)]K(t)\\
 & - 2\big(\bar{G}(t)C(t)^{\prime}D(t)^{\prime}+\bar{G}(t) B(t)^{\prime}+S(t)^{\prime} \big) K(t).
\end{align*}
We need the following property of the ODEs (\ref{def_hatGt}) and (\ref{def_barGt}).
\begin{lemma}\label{lem_nonnegative}
Under the Assumption \ref{asmp_psd}, the solutions of the ODEs (\ref{def_hatGt}) and (\ref{def_barGt}) satisfy $\hG(t) \geq 0$ and $\bG(t)\geq 0$ for all $t\in [0,T]$.
\end{lemma}
\begin{proof}
To prove $\hG(t)\geq 0$, we introduce another ODE as follows,
\begin{align}
\dot{\hat{G}}^*(t)&= -\hat{G}^*(t)C(t)^{\prime}C(t) - 2\hat{G}^*(t)A(t) - q(t) - \min_{K(t)} \hat{f}(t,K(t),\hat{G}^*(t)) \label{Gstar}
\end{align}
where $\hat{G}^*(t)$ is the unknown function. Solving the optimization problem in the right-hand side of the equation (\ref{Gstar}) yields
\begin{align}
\dot{\hat{G}}^*(t)&= -\hat{G}^*(t)C(t)^{\prime}C(t) - 2\hat{G}^*(t)A(t) - q(t) \notag\\
& -\big(\hG^*(t)D(t)C(t) + \hG^*(t)B(t) + S(t)\big)^{\prime} \notag\\
&\times\big(\hG^*(t)D(t)D(t)^{\prime} + R(t)\big)^{-1} \notag\\
&\times \big(\hG^*(t)D(t)C(t) + \hG^*(t)B(t) + S(t)\big). \label{Gstar2}
\end{align}
Note that the difference between (\ref{def_hatGt}) and (\ref{Gstar}) is that there is no control constraint in (\ref{Gstar}). Furthermore, it is not hard to see that the equation (\ref{Gstar2}) is the correspondent Riccati equation of the classical unconstrained LQ optimal control model. It has been known that the solution of (\ref{Gstar2}) satisfies $\hG^*(t)$$\geq$$0$ for all $t\in[0,T]$(see, e.g., Chapter 7 in\cite{YongZhou:1999}). Due to the constraint set, the right-hand side of equation (\ref{Gstar}) is no larger than the right-hand side of equation (\ref{def_hatGt}). Thus, it has $|\dot{\hat{G}}^*(t)|$ $\leq $
$|\dot{\hat{G}}(t)|$ for all $t$ $\in$ $[0,T]$. Taking the integration on both sides gives the $0\leq $ $ \hat{G}^*(t)$ $\leq $ $\hat{G}(t)$ for all $t\in [0,T]$. This proof can be also applied to the equation (\ref{def_barGt}). We omit the detail.
\end{proof}

Now, we provide the main result for problem ($\cP_{\LQ}^T$).
\begin{theorem} \label{thm_PLQ}
The optimal control policy of problem ($\cP_{\LQ}^{T}$) is
\begin{align}
u^*(t)=
\begin{cases}
~~~\hat{K}^*(t) x(t)~~&\textrm{if}~~x(t) \geq 0\\
-\bar{K}^*(t) x(t)~~&\textrm{if}~~x(t) < 0\\
\end{cases} \label{thm_PLQ_ut}
\end{align}
for $t\in[0,T]$, where $\hat{K}^*(t)$ and $\bar{K}^*(t)$ are defined as,
\begin{align}
\hat{K}^*(t)&=\arg \min_{K(t) \in \K(t)} \hat{f}(t,K(t),\hat{G}(t)), \label{def_hK}\\
\bar{K}^*(t)&=\arg \min_{K(t) \in \K(t)} \bar{f}(t,K(t),\bar{G}(t)), \label{def_bk}
\end{align}
respectively, and $\hat{G}(t)$ and $\bar{G}(t)$ are defined in (\ref{def_hatGt}) and (\ref{def_barGt}), respectively. Moreover, the optimal objective value for problem $(\cP_{\LQ}^T)$ is
\begin{align}
v(\cP_{\LQ}^T)=x(0)^2 \Big(\hat{G}(0) 1_{\{x(0)\geq0\}} + \bar{G}(0) 1_{\{x(0)<0\}} \Big)
\label{thm_PLQ_v0}
\end{align}
\end{theorem}
\begin{proof}
The structure of the proof is as follows. We verify that the following value function,
\begin{align}
V(t,x)=x^2 \Big(\hat{G}(t) 1_{\{x\geq0\}} + \bar{G}(t) 1_{\{x<0\}} \Big)\label{thm_HJB_vt}
\end{align}
is the viscosity solution for the HJB equation (\ref{def_HJB}) and the associated optimal feedback control is,
\begin{align}
u^{\dag}(t,x)=
\begin{cases}
~~\hat{K}^*(t) x~~&\textrm{if}~~x \geq 0,\\
-\bar{K}^*(t) x~~&\textrm{if}~~x < 0.\\
\end{cases} \label{thm_HJB_ut}
\end{align}
Substituting the dummy variable $x$ by state variable $x(t)$ in (\ref{thm_HJB_vt}) and (\ref{thm_HJB_ut}), we can obtain the optimal control policy (\ref{thm_PLQ_ut}) and optimal objective value (\ref{thm_PLQ_v0}) for problem $(\cP_{\LQ}^T)$, respectively.

To show that (\ref{thm_HJB_vt}) is the viscosity solution of (\ref{def_HJB}), we first consider the region $\Phi_1$ in the $(t,x)$-plane as,
\begin{align*}
\Phi_1=\Big\{(t,x)\in[0,T]\times \bR ~\big|~ x > 0  \Big\}.
\end{align*}
Obviously, in region $\Phi_1$, the value function (\ref{thm_HJB_vt}) becomes $V(t,x)$$=$$x^2 \hat{G}(t)$, which enables us to compute the derivatives as,
\begin{align*}
V_t(t,x)=\dot{\hat{G}}(t)x^2,~~V_x(t,x)=2\hat{G}(t)x,~~V_{xx}(t,x)=2\hat{G}(t).
\end{align*}
We then substitute these terms to HJB equation (\ref{def_HJB}). The left-hand side(LHS) of (\ref{def_HJB}) becomes
\begin{align}
\textrm{LHS}=&\dot{\hat{G}}(t)x^2 + \hat{G}(t)C(t)^{\prime}C(t)x^2 + 2\hat{G}(t)A(t)x^2 \notag\\
& + q(t)x^2 + \min_{u\in \cU_t(x)} L(x,u,t), \label{thm_HJB_Phi1_LHS}
\end{align}
where $L(x,u,t)$ is defined as,
\begin{align*}
L(x,u,t)&= u^{\prime}\big(\hat{G}(t)D(t)D(t)^{\prime}+R(t)\big)u \notag\\
& + 2x\big(\hat{G}(t)C(t)^{\prime}D(t)^{\prime}+\hat{G}(t) B(t)^{\prime}+S(t)^{\prime} \big)u.
\end{align*}
Note that, from Assumption \ref{asmp_psd} and Lemma \ref{lem_nonnegative}, it has $R(t)$$+$$\hat{G}(t)D(t)D(t)^{\prime}$$\succ$$0$. Thus, the inner optimization problem,  $\min_{u\in \cU_t(x)}$ $L(x,u,t)$, in  equation (\ref{thm_HJB_Phi1_LHS}) is a convex optimization problem, which admits the unique optimal solution. Now we apply the Theorem \ref{thm_sep} to this problem. More specifically, we regard terms $R(t)$$+$$\hat{G}(t)D(t)D(t)^{\prime}$, $\hat{G}(t)C(t)^{\prime}D(t)^{\prime}$ $+$ $\hat{G}(t) B(t)^{\prime}$$+$$S(t)^{\prime}$ and $x$ as the terms $\Omega$, $\omega$ and $\alpha$ in problem $(\cP(\alpha))$, respectively. Thus, in region $\Phi_1$, it has
\begin{align*}
 x\hK^*(t) = \arg \min_{u \in \cU_t(x)} L(x,u,t),
\end{align*}
where $\hK^*(t)$ is defined in (\ref{def_hK}). From (\ref{thm_sep_obj}), the LHS (\ref{thm_HJB_Phi1_LHS}) can be written as,
\begin{align*}
\textrm{LHS}&=x^2\Big(\dot{\hat{G}}(t) + \hat{G}(t)C(t)^{\prime}C(t)+ 2\hat{G}(t)A(t)\notag\\
& + q(t)+ \min_{K(t) \in \K(t)} \hat{f}(t,K(t),\hat{G}(t))\Big).
\end{align*}
From the definition (\ref{def_hatGt}), it has $\textrm{LHS}$ $=$ $0$, which implies that the value function (\ref{thm_HJB_vt}) is the solution of the HJB equation (\ref{def_HJB}) in region $\Phi_1$ with the optimal control police being $x\hK^*(t)$.

Now we consider the region $\Phi_2$ in the $(t,x)$-plane as,
\begin{align*}
\Phi_2=\Big\{(t,x)\in[0,T]\times \bR ~\big|~ x < 0 \Big\}.
\end{align*}
In this region, the value function becomes $V(t,x)$ $=$ $x^2 \bG(t)$. We can apply the similar solution procedure for region $\Phi_2$ to prove that the value function, $V(t,x)$ $=$ $x^2\bG(t)$, is the solution of the HJB equation (\ref{def_HJB}) and derive the associated optimal control policy as (\ref{thm_HJB_ut}). The detail is omitted.

Finally, we consider the region $\Phi_3$ in $(t,x)$-plane as,
\begin{align*}
\Phi_3=\Big\{(t,x)\in[0,T]\times \bR ~|~ x = 0  \Big\}.
\end{align*}
The nonsmoothness property of $V(t,x)$ arises in this region. However, the value function $V(t,x)$ is still continuous for any $(t,x) \in \Phi_3$, since it has $V(t,x)$ $=$ $\hat{G}(t)x^2$ $=$ $\bar{G}(t)x^2$ $=$$0$ in this region. In addition, the first order partial derivative also exists for $V(t,x)$ at any points in region $\Phi_3$, i.e., it has $V_t(t,x)$ $=$ $\dot{\hat{G}}x^2$ $=$ $\dot{\bar{G}}x^2$ $=$ $0$ and $V_x(t,x)$$=$$2{\hat{G}}x$$=$$2{\bar{G}}x$$=0$ for $(t,x)$$\in$$\Phi_3$. However, due to $\hat{G}(t)$$\neq$$\bar{G}(t)$, the second order partial derivative, $V_{xx}(t,x)$, does not exist at any points on $\Phi_3$, which causes the nonsmoothness of $V(t,x)$. For such reason, we need to use the framework of viscosity solutions to conquer this difficulty. We follow the similar idea of the verification theorem given in \cite{ZhouYongLi:1997} and \cite{LiZhouLim:2002} to show that the value function (\ref{thm_HJB_vt}) is the viscosity solution for the HJB equation (\ref{def_HJB}). According to the theory of viscosity solutions(see \cite{YongZhou:1999} \cite{ZhouYongLi:1997}), we introduce the second-order superdifferential and subdifferential of $V(t,x)$ at $(t,x)$ $\in$ $\Phi_3$ as follows,\footnote{The detail definition of the second-order supperdifferential and subdifferential can be found in Chp. 5 in reference \cite{YongZhou:1999}}
\begin{align*}
\begin{cases}
D_{t,x}^{1,2,+}V(t,x)=\{0\} \times \{0\} \times [\hat{G}(t),+\infty),\\
D_{t,x}^{1,2,-}V(t,x)=\{0\} \times \{0\} \times (-\infty,\bar{G}(t)],
\end{cases}
\end{align*}respectively. Moreover, we define the following term for the HJB equation (\ref{def_HJB}),
\begin{align*}
W(t,x,u,p,P)
&=\frac{1}{2}P\big(C(t)C(t)^{\prime}x^2+2xC(t)^{\prime}D(t)^{\prime}u+u^{\prime}D(t)D(t)^{\prime}u\big)\\
&+p\big(A(t)x+B(t)^{\prime}u\big) + \big( q(t)x^2+2xS(t)^{\prime}u+u^{\prime}R(t)u \big).
\end{align*}
For any $(\eta,p,P)$ $\in$ $D_{t,x}^{1,2,+}V(t,x)$ with $(t,x)$ $\in$ $\Phi_3$, it has
\begin{align*}
& \eta+\min_{u \in \cU_t(x)} W(t,x,u,p,P)\\
= & \min_{u \in \cU_t(x)} \Big\{ \frac{1}{2}u^{\prime}\big(P D(t)D(t)^{\prime} + 2R(t) \big)u \Big\}\\
\geq & \min_{u \in \cU_t(x)} \Big\{ \frac{1}{2}u^{\prime}\big(\hat{G}(t) D(t)D(t)^{\prime} + 2R(t) \big) u \Big\}=0.
\end{align*}
Hence, we can seen that $V(t,x)$ is a viscosity subsolution for the HJB equation (\ref{def_HJB}).\footnote{The definition of the subsolution and supsolution can be found in \cite{YongZhou:1999}.} On the other hand, for any $(\eta,p,P)$ $\in$ $D_{t,x}^{1,2,-}V(t,x)$ with $(t,x)$ $\in$ $\Phi_3$, it has,
\begin{align*}
& \eta+\min_{u \in \cU_t(x)} W(t,x,u,p,P)\\
=& \min_{u \in \cU_t(x)} \Big\{ \frac{1}{2}u^{\prime}\big(P D(t)D(t)^{\prime} + 2R(t) \big)u \Big\}\\
\leq & \min_{u \in \cU_t(x)} \Big\{ \frac{1}{2}u^{\prime}\big(\bar{G}(t) D(t)D(t)^{\prime} + 2R(t) \big)u \Big\}=0,
\end{align*}
which shows that $V(t,x)$ is a viscosity supersolution of the HJB equation (\ref{def_HJB}). Overall, we can obtain that the value function $V(t,x)$ given in (\ref{thm_HJB_vt}) is the viscosity solution for the HJB equation (\ref{def_HJB}) with terminal condition $V(T,x)$$=$$q_Tx^2$. Note that, for any $(t,x)$ $\in$ $\Phi_3$, taking $(\eta^*(t,x),p^*(t,x),P^*(t,x))$ $=$ $(0,0,\hat{G}(t))$ $\in$ $D_{t,x}^{1,2,+}V(t,x)$ and $u^{\dag}(t,x)=0$, it has
\begin{align*}
\eta^*(t,x)+W(t,x,u^{\dag}(t,x),p^*(t,x),~P^*(t,x))=0.
\end{align*}
Therefore, based on the verification theorem in Theorem 3.1 of \cite{ZhouYongLi:1997}, we obtain that $u^{\dag}(t,x)$ defined in expression (\ref{thm_HJB_ut}) is the optimal feedback control.
\end{proof}

In Theorem \ref{thm_PLQ}, the two Raccati equations (\ref{def_hatGt}) and (\ref{def_barGt}) plays the same role as the single Riccati equation induced from the classical unconstrained continuous-time stochastic LQ control problem. If there is no constraints $\cU_t(x(t))$, these two equations will merge to a single equation, which is the one generated from the classical LQ control problem. To solve the equations (\ref{def_hatGt}) and (\ref{def_barGt}), usually, we need to use the numerical methods, since there is no analytical solution for the inner optimization problems in both equations for a general linear constraints. More specifically, we can discretize the time horizon $[0,T]$ in to some discrete time points. We then approximate these two equations by the difference formula according to these discrete sample periods. At any fixed period, we could solve the inner optimization problem and achieve the solution of these equations.

\section{Optimal Solution for Problem (\texorpdfstring{$\cP_{\LQ}^{\infty}$}{Lg})}\label{se_sol_Pinf}
In this section, we focus on developing the stationary control policy for the problem ($\cP_{\LQ}^{\infty}$). The main idea is as follows. We first study the associated problem with a finite-horizon and then investigate the asymptotic performance by extending the time horizon to infinity. Introduce the following auxiliary problem with finite horizon $T$,
\begin{align}
(\cA^{T})&~\min_{\{u(t)\}_{t=0}^T}~\E \left[ \int_{0}^{T}
\left(\begin{array}{cc}
u(t)\\
x(t)
\end{array}\right)^{\prime}Q
\left(\begin{array}{cc}
u(t)\\
x(t)
\end{array}\right)dt~\right] \label{def_A_T}\\
{\rm s.t.}~~&~\begin{cases}
dx(t)=\big( Ax(t)+B^{\prime}u(t)\big )dt\\
~~~~~~~~~~+\big( x(t)C^{\prime}+u(t)^{\prime}D \big)dW(t), \\
x(0)=x_0,\\
u(t)\in \tilde{\cU}_t(x(t)), ~~t\in[0,T].
\end{cases} \notag
\end{align}
Obviously, problem ($\cA^{T}$) is just a special case of problem ($\cP_{\LQ}^{T}$), if we set the the all parameters in problem ($\cP_{\LQ}^T$) as the time invariant constants and set $q_T=0$. To investigate the property of problem ($\cA^T$) for different time horizon length, we introduce the notations $\hat{G}(t;T)$ and $\bar{G}(t;T)$ to denote the outputs(solutions) of ODEs (\ref{def_hatGt}) and (\ref{def_barGt}) at time $t$ with the horizon being $T$ and the terminal conditions being $\hat{G}(T;T)$ $=$ $\bar{G}(T;T)$$=$ $0$, respectively.

Based on Theorem \ref{thm_PLQ}, the optimal objective value of the auxiliary problem ($\cA^{T}$) can be expressed as,
\begin{align}
v(\cA^T)=x_0^2 \Big( \hat{G}(0;T) \1_{\{x_0\geq0\}} + \bar{G}(0;T) \1_{\{x_0<0\}} \Big). \label{thm_ALQ_v0}
\end{align}
We can see that the optimal value of problem $(\cA^T)$ is either $x_0^2\hat{G}(0;T)$ or $x_0^2\bar{G}(0;T)$, which depends on the sign of the initial state $x_0$. In addition, the outputs $\hat{G}(t;T)$ and $\bar{G}(t;T)$ have the following property.

\begin{lemma} \label{lem_G_homo}
Consider the problems ($\cA^{T}$) and ($\cA^{T+\tau}$) with $\tau>0$. Let \{$\hat{G}(t;T)$, $\bar{G}(t;T)$\} and \{$\hat{G}(t;T+\tau)$, $\bar{G}(t;T+\tau)$\} be the solutions of ODEs (\ref{def_hatGt}) and (\ref{def_barGt}) at time $t$ for these two problems, respectively. It has $\hat{G}(t;T)$ $=$ $\hat{G}(t+\tau;T+\tau)$ and $\bar{G}(t;T)$ $=$ $\bar{G}(t+\tau;T+\tau)$, for any $t\in[0,T]$.
\end{lemma}
Lemma \ref{lem_G_homo} can be easily verified by applying Theorem \ref{thm_PLQ} to problems $(\cA^T)$ and $(\cA^{T+\tau})$. Lemma \ref{lem_G_homo} shows that outputs $\hat{G}(t;T)$ and $\bar{G}(t;T)$ are time homogeneous, that is, these two functions, $\hat{G}(t;T)$ and $\bar{G}(t;T)$ only depend on time difference $T-t$.

Introduce the following set,
$$\tK:=\{K\in \bR^n~|~HK\leq d~\},$$
which is associated with the constraint set $\tilde{\cU}_t(x(t))$. We then present the main result of problem ($\cP_{\LQ}^{\infty}$).
\begin{theorem} \label{thm_PLQ_inf}
For auxiliary problem $(\cA^T)$, if $\hG(0;T)$ and $\bG(0;T)$ have a uniformed upper bound $M>0$ for any $T$, then there exists limits $\hG^*>0$ and $\bG^*>0$ such that $\hG^*$ $=$ $\lim_{T\rightarrow \infty} \hG(0;T)$ and $\bG^*$ $=$ $\lim_{T\rightarrow \infty} \bG(0;T)$. Moreover, $\hG^*$ and $\bG^*$  satisfy the following equations,
\begin{align}
0=& -\hat{G}^*C^{\prime}C - 2\hat{G}^*A - q - \min_{K \in \tK} \hat{f}(K,\hat{G}^*), \label{def_equ_hG}\\
0=& -\bar{G}^*C^{\prime}C - 2\bar{G}^*A - q - \min_{K \in \tK} \bar{f}(K,\bar{G}^*), \label{def_equ_bG}
\end{align}
where $\hat{f}(K,\hat{G}^*)$ and $\bar{f}(K,\bar{G}^*)$ are defined as,
\begin{align}
\hat{f}(K,\hat{G}^*)=&K^{\prime}\big(\hat{G}^*DD^{\prime}+ R\big)K + 2\big(\hat{G}^*C^{\prime}D^{\prime} +\hat{G}^*B^{\prime}+S^{\prime} \big) K, \label{def_equ_hf}\\
\bar{f}(K,\bar{G}^*)=&K^{\prime}\big(\bar{G}^*DD^{\prime}+ R\big)K - 2\big(\bar{G}^*C^{\prime}D^{\prime} +\bar{G}^*B^{\prime}+S^{\prime} \big) K, \label{def_equ_bf}
\end{align}
respectively. The optimal control policy $u^*(t)$ for problem $(\cP_{\LQ}^{\infty})$ is
\begin{align}
u^*(t)=
\begin{cases}
~~~\hat{K}^* x(t)~~\textrm{if}~~x(t) \geq 0\\
-\bar{K}^* x(t)~~\textrm{if}~~x(t) < 0\\
\end{cases} \label{thm_PLQ_ut_inf}
\end{align}
with $\hat{K}^*$ and $\bar{K}^*$ given by,
\begin{align}
\hat{K}^*=&\arg \min_{K \in \tK} \hat{f}(K,\hat{G}^*)  \label{def_equ_hK}\\
\bar{K}^*=&\arg \min_{K \in \tK} \bar{f}(K,\bar{G}^*). \label{def_equ_bK}
\end{align}
Moreover, the optimal objective value for problem $(\cP_{\LQ}^{\infty})$ is given by,
\begin{align}
v(\cP_{\LQ}^{\infty})=x(0)^2 \big( \hat{G}^* 1_{\{x(0)\geq 0\}} + \bar{G}^* 1_{\{x(0)<0\}} \big).
\label{thm_PLQ_v0_inf}
\end{align}
Furthermore, under the stationary optimal control policy (\ref{thm_PLQ_ut_inf}), the closed-loop system (\ref{def_xt_inf}) is $L^2$-asymptotically stable, i.e.,
\begin{align*}
\lim_{t\rightarrow\infty}\E[(x^*(t))^2]=0.
\end{align*}
\end{theorem}
\begin{proof}
We first show that the optimal objective value of problem ($\cA^{T}$) is nondecreasing as the time horizon $T$ increases. We assume that the policy $\tilde{u}(\cdot)$ solves problem ($\cA^{T+\tau}$) and $\tilde{x}(\cdot)$ is the associated state trajectory, which leads to the following inequality,
\begin{align}
&v(\cA^{T+\tau})=\E \left[ \int_{0}^{T+\tau}
\left(\begin{array}{cc}
\tilde{u}(t)\\
\tilde{x}(t)
\end{array}\right)^{\prime}Q
\left(\begin{array}{cc}
\tilde{u}(t)\\
\tilde{x}(t)
\end{array}\right)dt \right] \notag\\
&~~\geq \E \left[ \int_{0}^{T}
\left(\begin{array}{cc}
\tilde{u}(t)\\
\tilde{x}(t)
\end{array}\right)^{\prime}Q
\left(\begin{array}{cc}
\tilde{u}(t)\\
\tilde{x}(t)
\end{array}\right)dt \right]\geq v(\cA^{T}), \label{thm_PLQinf_ineq}
\end{align}
where the last inequality is from the fact that the truncated policy $\tilde{u}(t)|_{t=0}^T$ is just a feasible policy of problem $(\cA^T)$. According to expressions (\ref{thm_ALQ_v0}) and (\ref{thm_PLQinf_ineq}), it has $\hat{G}(0;T+\tau)$ $\geq $ $\hat{G}(0;T)$ and $\bar{G}(0;T+\tau)$ $\geq $ $\bar{G}(0;T)$. Thus, it is not hard to see that both $\hat{G}(0;T)$ and $\bar{G}(0;T)$ are nondecreasing as $T$ increases. When there exists a uniformed upper bound $M$ such that $\hat{G}(0;T)<M$ and $\bar{G}(0;T)<M$ for any $T$, the nondecreasing sequences $\hat{G}(0;T)$ and $\bar{G}(0;T)$ will converge as $T$ goes to infinity. Next, we show that the limits $\hG^*$ and $\bG^*$ satisfy the equation (\ref{def_equ_hG}) and (\ref{def_equ_bG}). Introduce the following equations for $t\in[0,\infty)$:
\begin{align}
\dot{\hat{G}}^{\dag}(t)&= -\hat{G}^{\dag}(t)C^{\prime}C - 2\hat{G}^{\dag}(t)A - q - \min_{K \in \tK} \hat{f}\big(t,K,\hat{G}^{\dag}(t) \big), \label{thm_PLQ_inf_hGd}\\
\dot{\bar{G}}^{\dag}(t)&= -\bar{G}^{\dag}(t)C^{\prime}C - 2\bar{G}^{\dag}(t)A - q - \min_{K \in \tK} \bar{f}\big(t,K,\bar{G}^{\dag}(t)\big), \label{thm_PLQ_inf_bGd}
\end{align}
where $\hG^{\dag}(0)=0$ and $\bG^{\dag}(0)=0$ with $\hat{f}(t,K,\hat{G}^{\dag}(t))$ and $\bar{f}(t,K,\bar{G}^{\dag}(t))$ defined as follows,
\begin{align*}
\hat{f}(t,K,\hG^{\dag}(t))&=K^{\prime}\big( \hG^{\dag}(t)DD^{\prime}+ R \big)K + 2\big( \hG^{\dag}(t)C^{\prime}D^{\prime}+\hG^{\dag}(t) B^{\prime}+S^{\prime} \big) K, \\
\bar{f}(t,K,\bG^{\dag}(t))&=K^{\prime}\big( \bG^{\dag}(t)DD^{\prime}+ R \big)K - 2\big( \bG^{\dag}(t)C^{\prime}D^{\prime}+\bG^{\dag}(t) B^{\prime}+S^{\prime} \big) K.
\end{align*}
From Lemma \ref{lem_G_homo}, we have the following result for any $T>0$,
\begin{align*}
\hat{G}(t;T)=\hat{G}^{\dag}(T-t),~~\bar{G}(t;T)=\bar{G}^{\dag}(T-t),
\end{align*}
for $t\in[0,T]$. Taking $t=0$ in the above equations yields
\begin{align*}
\hat{G}(0;T)=\hat{G}^{\dag}(T),~~\bar{G}(0;T)=\bar{G}^{\dag}(T),
\end{align*}
for $T\geq0$. That is to say, studying the asymptotic properties of $\hG(0;T)$ and $\bG(0;T)$ with respect to $T$ is equivalent to study  $\hG^{\dag}(T)$ and $\bG^{\dag}(T)$ with respect to $T$. Since the limits of  $\hG(0;T)$ and $\bG(0;T)$ exists, it has $\lim_{t\rightarrow \infty}\dot{\hG}^{\dag}(t) $ $=$ $0$ and $\lim_{t\rightarrow \infty}\dot{\bG}^{\dag}(t) $ $=$ $0$. Thus, the equations (\ref{thm_PLQ_inf_hGd}) and (\ref{thm_PLQ_inf_bGd}) become (\ref{def_equ_hG}) and (\ref{def_equ_bG}), respectively.

Finally, we show that the closed-loop system (\ref{def_xt_inf}) is asymptotically stable under the stationary optimal policy $u^*(t)$. When $x(t)\geq 0$, applying the formulations $(9)$ in \cite{ChenZhou:2004} yields to the following expression with stationary optimal policy (\ref{thm_PLQ_ut_inf}),
\begin{align*}
dx(t)^2&=\Big[2A+C^{\prime}C+2(B^{\prime}+C^{\prime}D^{\prime})\hat{K}^*+(\hat{K}^*)^{\prime}DD^{\prime}\hat{K}^* \Big]\\
& \times x(t)^2dt + (2C^{\prime}+2(\hat{K}^*)^{\prime}D)x(t)^2dW(t).
\end{align*}
Taking the expectation of the above SDE gives rise to
\begin{align*}
\frac{d\E[x(t)^2]}{dt}=\hat{N}(\hat{K}^*)\E[x(t)^2]
\end{align*}
where
\begin{align*}
\hat{N}(K)=\big(2A+C^{\prime}C+2(B^{\prime}+C^{\prime}D^{\prime})K+(K)^{\prime}DD^{\prime}K \big).
\end{align*}
Comparing the above equation with (\ref{def_equ_hG}) yields to,
\begin{align*}
G^*\hat{N}(\hat{K}^*)&=-[q+(\hat{K}^*)^{\prime}R\hat{K}^*+2S^{\prime}\hat{K}^*]\\
&=-\left(\begin{array}{cc}
\hat{K}^*\\
1
\end{array}\right)^{\prime}Q
\left(\begin{array}{cc}
\hat{K}^*\\
1
\end{array}\right).
\end{align*}
From the facts that $Q$$\succ$$0$ and $G^*$$>$$0$, we obtain $\hat{N}(\hat{K}^*)<0$, which further  leads to
$$\lim_{t\rightarrow +\infty}\E[x^*(t)^2 \1_{\{x^*(t)\geq0\}}]=0.$$
Similarly, when $x(t)<0$, we can prove that $\bar{N}(\bar{K}^*)$ $<$ $0$ with
\begin{align*}
\bar{N}(K)&=\big(2A+C^{\prime}C-2(B^{\prime}+C^{\prime}D^{\prime})K+(K)^{\prime}DD^{\prime}K \big).
\end{align*}
Under this case, it has $\lim_{t\rightarrow +\infty}\E[x^*(t)^2 \1_{\{x^*(t)<0\}}]=0$. Overall,  it has $\lim_{t\rightarrow +\infty}$ $\E[x^*(t)^2]$ $=$ $0$, which completes the proof.
\end{proof}

In Theorem \ref{thm_PLQ_inf}, we may employ the numerical method to solve the equations  (\ref{def_equ_hG}) and (\ref{def_equ_bG}). Define the following functions associated with equation (\ref{def_equ_hG}) and (\ref{def_equ_bG}),
\begin{align}
\hat{F}(\hG)=& -\hG C^{\prime}C - 2\hG A - q - \min_{K \in \tK} \hat{f}(K, \hG),\label{def_hF}\\
\bar{F}(\bG)=& -\bG C^{\prime}C - 2\bG A - q - \min_{K \in \tK} \bar{f}(K, \bG).\label{def_bF}
\end{align}
Obviously, solving equations $\hat{F}(\hG)=0$ and $\bar{F}(\bG)=0$ provides the solution $\hG^*$ and $\bG^*$. We take equation (\ref{def_equ_hG}) as an example to illustrate the main idea of Newton's iterative method. Let $\hG^k>0$ be the solution in $k$th iteration and the gradient be $ \nabla \hat{F}(\hG^k)  $. Note that the following equation,
\begin{align*}
y-\hat{F}(\hG^k)= \nabla \hat{F}(\hG^k) (z -\hG^k ),
\end{align*}
describes the straight line passing through point $(\hG^k, \hat{F}(\hG^k))$ with the slop $\nabla \hat{F}(\hG^k)$ in space of $\{y,z\}$.\footnote{The gradient of $\hat{F}(\hG^k)$ can be approximated as, $\nabla \hat{F}(\hG^k)$ $\approx$
$\frac{\hat{F}(\hG^k+\epsilon) -\hat{F}(\hG^k)}{\epsilon}$ for some number $\epsilon$.} This equation is a linear approximation of the original equation $\hat{F}(\hG)=0$ at point $\hG^k$. Letting $y=0$  gives intersection of this line with $0$ and provides us the solution of the next iteration, $\hG^{k+1} = \hG^k - \hat{F}(\hG^k)/ \nabla \hat{F}(\hG^k)$. The solution of equation $\hat{F}(\hG)=0$ can be obtained by repeating this procedure until some criteria is satisfied. We summarize this procedure in Algorithm \ref{alg:hGbG}.
\begin{algorithm}[htb] \label{alg:hGbG}
\caption{Iterative solution method for (\ref{def_equ_hG})}
\label{alg:hGbG}
\begin{algorithmic}[1]
\Require
The parameters $A$, $B$, $C$, $D$, $R$, $S$ and $q$, the initial solutions $\hG^{\dagger}$, the parameter $\epsilon>0$ and the maximum number of the iteration, $I_{\max}$.
\Ensure
\State Let $k\leftarrow 0$,  $\hG^k\leftarrow \hG^{\dagger}$; \label{code:step 1}
\State Calculate $\hF(\hG^k)$;  \label{code:step 2}
\State If($|\hF(\hG^k)| < \epsilon$), stop, return the solution $\hG^k$; Otherwise go to next step;
\State Calculate $\nabla \hF(\hG^k)$ and
\begin{align*}
\hG^{k+1} = \hG^k - \hat{F}(\hG^k)/ \nabla \hat{F}(\hG^k)
\end{align*}
\State  $k\rightarrow k+1$, if $k>I_{\max}$, go to next step, otherwise, got to Step \ref{code:step 2};
\State Fail to find the solution of Equations (\ref{def_equ_hG}). \label{code:return_fail}
\end{algorithmic}
\end{algorithm}
The above procedure can be simply modified for equation (\ref{def_equ_bG}) by replacing $\hat{F}(\hG)$ by $\bar{F}(\bG)$.

\section{Application in Dynamic Mean-Variance Portfolio Selection}
\label{se_application_MV}
In this section, we use the solution method provided in Section \ref{se_sol_P} to solve the constrained dynamic MV portfolio selection problem. We consider a financial market with $n$ risky assets and one risk free asset. The price of the risk free asset is $S_0(t)$, which is governed by the following equation,
\begin{align*}
\begin{cases}
dS_0(t)=r(t)S_0(t)dt,~~t \in [0,T],\\
S_0(0)=s_0>0,
\end{cases}
\end{align*}
where $r(t)>0$ is the risk free return rate. The price of $n$ risky assets satisfy the SDE,
\begin{align*}
\begin{cases}
dS_i(t)=S_i(t)\{\mu_i(t)dt+\sum_{j=1}^{n}\sigma_{ij}(t)dW^j(t)\},\\
S_i(0)=s_i>0
\end{cases}
\end{align*}
where $\mu_i(t)>0$ and $\sigma_{ij}(t)$ are the appreciation rate and volatility coefficient, respectively. Let $\mu(t)$ $:=$ $\big(\mu_1(t),\cdots,\mu_n(t)\big)^{\prime}$ and
\begin{align*}
\sigma(t):=\left[
  \begin{array}{ccc}
    \sigma_{11}(t) & \cdots & \sigma_{1n}(t)\\
       \cdots      & \cdots &    \cdots\\
    \sigma_{n1}(t) & \cdots & \sigma_{nn}(t)
  \end{array}
\right],
\end{align*}
for all $t\in [0,T]$.  We assume that $r(t)$, $\mu(t)$ and $\sigma(t)$ are the deterministic functions of time and they are bounded for $t\in [0,T]$. We also assume the following  nondegeneracy condition holds for some constant $\delta>0$,
\begin{align*}
\sigma(t)\sigma(t)^{\prime} \geq  \delta \I, ~~t\in [0,T].
\end{align*}
An investor enters the market with initial wealth $x_0$ and allocates this wealth continuously during the invest horizon $[0,T]$. The total wealth at time $t$ is denoted as $x(t)$ and the portfolio decision vector is denoted as $ u(t)$$:=$$\big(u_1(t),\cdots,u_n(t)\big)^{\prime}$, which represents the allocations of the wealth in $n$ risky assets. The wealth level $x(t)$ evolves according to the following SDE (e.g., see Zhou $et~al.$ \cite{ZhouLi:2000}),
\begin{align}
\begin{dcases}
dx(t)=\Big( r(t)x(t)+b(t)^{\prime}u(t)\Big) dt+u(t)^{\prime}\sigma(t)dW(t),~~~~t\in[0,T],\\
x(0)=x_0>0,
\end{dcases}\label{MV_x}
\end{align}
where $b(t)$ $:=$ $\big( b_1(t),\cdots,b_n(t)\big)^{\prime}$ $=$ $\big(\mu_1(t)-r(t),\cdots,\mu_n(t)-r(t)\big)^{\prime}$ is the excess return rate vector. We also assume the following condition holds.
\begin{assumption} \label{asmp_b}
 There exists some $i^*\in \{1,\cdots,n\}$ such that  $b_{i^*}(t)>0$.
\end{assumption}
The Assumption \ref{asmp_b} says that the return rate of some risky assets should be greater than the return rate of the risk free asset, which is reasonable in real market.

Motivated by the restrictions in real investment, we consider the following constraint for the portfolio decision vector,
\begin{align}
H(t)u(t) \geq \0,~~t \in [0,T], \label{MV_u_constraint}
\end{align}
where $H(t)\in \bR^{k \times n}$ is the deterministic matrix for $t\in [0,T]$. The above constraint is known as the convex cone constraint, which includes various portfolio constraints as its special cases, e.g., the no-shorting constraint(\cite{LiZhouLim:2002}). In the following part, we use $\Var[x]$$:=$$\E[(x-\E[x])^2]$ to denote the variance of some random variable $x$. The investor adopts the following dynamic MV portfolio decision model to guide his investment,
\begin{align*}
(MV):~\min_{\{u(t)\}|_{t=0}^{T}}&~~\Var[x(T)]+\int_{0}^{T}\E[u(t)^{\prime}R(t)u(t)]dt\\
(s.t.)
&~~\E[x(T)]=d\\
&~~\{x(t),u(t)\}~~\textrm{satisfies (\ref{MV_x}) and (\ref{MV_u_constraint})}
\end{align*}
where $d$ is  expected terminal wealth level. Usually, we set $d>x_0e^{\int_{0}^{T}r(s)ds}$, i.e., the target expected terminal wealth should be greater than the one obtained by investing all into the risk free account. In model (MV), the penalty term $u(t)^{\prime}R(t)u(t)$ is used to controlled the risk exposure in the risky assets. It is worthwhile to mention that our model is more general than the one studied in the current literature, e.g., the work \cite{LiZhouLim:2002} only involves the no-shorting constraint \footnote{In \cite{LiZhouLim:2002}, some additional assumption is needed, e.g., they assume $\mu_i(t)>r(t)$ for all $i=1,\cdots, n$. This assumption is crucial to obtain their results. However, in our model, we can relax such assumption to Assumption \ref{asmp_b}.} and  \cite{HuZhou:2005} does not consider the penalty term $u(t)^{\prime}R(t)u(t)$ in the model.

To solve problem $(MV)$, we first reformulate it as a special case of the LQ control problem ($\cP_{\LQ}^T$). We utilize the embedding technique introduced by Li et al. \cite{LiNg:2000} to overcome difficulties of inseparability of variance term in the sense of dynamic programming. We consider the following auxiliary problem ($\widehat{MV}(\lambda)$) by introducing the Lagrange multiplier $\lambda \in \bR$ for constraint $\E[x(T)]=d$,
\begin{align*}
(\widehat{MV}(\lambda)):&~~\min_{\{u(t)\}|_{t=0}^{T}}~\E[(x(T)-\lambda)^2]+\int_{0}^{T}\E[u(t)^{\prime}R(t)u(t)]dt-(\lambda-d)^2\\
(s.t.)
&~~\{x(t),u(t)\}~~\textrm{satisfies (\ref{MV_x}) and (\ref{MV_u_constraint})}.
\end{align*}
We define the discount factor as
$$\rho(t):= e^{-\int_t^T r(s)ds}$$
for $t\in [0,T]$ and construct a new state variable $z(t)$ as $z(t):=x(t)-\lambda \rho(t)$. Replacing the state variable in $(\widehat{MV}(\lambda))$ yields the following problem,
\begin{align*}
(\overline{MV}(\lambda)):&\min_{\{u(t)\}|_{t=0}^{T}}\E[z(T)^2]+\int_{0}^{T}\E[u(t)^{\prime}R(t)u(t)]dt\\
(s.t.)
&~dz(t)=\big( r(t)z(t)+b(t)^{\prime}u(t)\big) dt\\
&~~~~~~~~~~~~~~+u(t)^{\prime}\sigma(t)dW(t).
\end{align*}
Obviously, this problem is a special case of problem $(\cP_{\LQ}^T)$ by setting $A(t)=r(t)$, $B(t)=b(t)$, $C(t)=0_{n\times 1}$,  $D(t)=\sigma(t)$, $q(t)=0$, $S(t)=0_{n\times 1}$, for $t\in[0,T)$ and $q_T=1$. Thus, similar to equations (\ref{def_hatGt}) and (\ref{def_barGt}), we introduce the following two ODEs for the two unknown functions $\hG^{mv}(\cdot): \bR \rightarrow \bR$ and $\bG^{mv}(\cdot): \bR \rightarrow \bR $,
\begin{align*}
\dot{\hG}^{mv}(t)&= - 2\hat{G}^{mv}(t)r(t) -\min_{ K:~H(t)K  \geq \0 } \hat{f}^{mv}(t,K,\hat{G}^{mv}(t)), \\
\dot{\bG}^{mv}(t)&= - 2\bar{G}^{mv}(t)r(t) - \min_{K:~ H(t)K \geq \0} \bar{f}^{mv}(t,K,\bar{G}^{mv}(t)),
\end{align*}
where $\hG^{mv}(T)$$=$$\bG^{mv}(T)=1$, respectively, and
\begin{align}
\hat{f}^{mv}(t,K, G)&= 2 G b(t)^{\prime} K + K^{\prime}\Big(G\sigma(t)\sigma(t)^{\prime}+R(t)\Big)K.\label{def_hf_mv}\\
\bar{f}^{mv}(t,K, G)&=-2 G b(t)^{\prime} K + K^{\prime}\Big(G\sigma(t)\sigma(t)^{\prime}+R(t)\Big)K. \label{def_bf_mv}
\end{align}
Furthermore, the functions $\hat{G}^{mv}(\cdot)$ and $\bar{G}^{mv}(\cdot)$ have the following properties.
\begin{lem}\label{lem_two_ineuqalities}
For any $t\in[0,T]$, the following inequalities hold true,
\begin{align}
&\hG^{mv}(t) \rho(t)^2 \leq 1,~~\label{MV_hG_convex}\\
&\bG^{mv}(t)\rho(t)^2 < 1.  \label{MV_bG_convex}
\end{align}
\end{lem}
\begin{proof} We first consider $\hG^{mv}(t)$. Since $K(t)=\0$ is a feasible solution of $\hat{f}^{mv}(t,K(t),\hat{G}^{mv}(t))$, we obtain,
\begin{align*}
\min_{K(t): H(t)K(t) \geq 0}&\hat{f}^{mv}(t,K(t),\hat{G}^{mv}(t)) \\
\leq&\hat{f}^{mv}(t,0,\hat{G}^{mv}(t))=0.
\end{align*}
Thus, it has $\dot{\hat{G}}^{mv}(t)$$\geq $$- 2\hat{G}^{mv}(t)r(t)$ which leads to the following inequality,
\begin{align*}
&\int_t^T \frac{1}{\hat{G}^{mv}(t)}d \hat{G}^{mv}(t) \geq - \int_t^T 2 r(s) ds\\
\Rightarrow &~~1-\hat{G}^{mv}(t)e^{-2\int_{t}^{T}r(s)ds} \geq 0,
\end{align*}
Similarly, we can prove that  $1-\bar{G}^{mv}(t)\rho(t)^2\geq 0$. Now, we show that
\begin{align*}
1-\bar{G}^{mv}(t) \rho(t)^{2}\not =0.
\end{align*}
If $1-\bar{G}^{mv}(t)\rho(t)^2$$=$$0$, it has,
\begin{align*}
\min_{H(t)K(t) \geq \0} \bar{f}^{mv}(t,K(t),\bar{G}^{mv}(t))=0,
\end{align*}
which is contradict to the fact
$$\min_{H(t)K(t) \geq \0} \bar{f}^{mv}(t,K(t),\bar{G}^{mv}(t))<0$$
according to expression (\ref{def_bf_mv}) and under the Assumption \ref{asmp_b}.
\end{proof}

Note that, in Lemma \ref{lem_two_ineuqalities}, the strict inequality only holds for (\ref{MV_bG_convex}). Such a result plays important role in the following part.

The following result characterizes the solution of problem (MV).
\begin{theorem} \label{thm_sol_MV} The associated optimal portfolio policy of problem $(MV)$ is given by,
\begin{align}
u^*(t)=
\begin{dcases}
\hat{K}^{mv}(t) (x(t)-\lambda^* \rho(t)) & \textrm{if}~x(t)-\lambda^* \rho(t) \geq 0,\\
-\bar{K}^{mv}(t)(x(t)-\lambda^* \rho(t)) & \textrm{if}~x(t)-\lambda^* \rho(t)  < 0,\\
\end{dcases} \label{mv_opt_u}
\end{align}
where
\begin{align}
\hat{K}^{mv}(t)&=\arg \min_{K:H(t)K \geq 0} \hat{f}^{mv}(t,K,\hat{G}^{mv}(t)), \label{def_hK_MV}\\
\bar{K}^{mv}(t)&=\arg \min_{K:H(t)K \geq 0} \bar{f}^{mv}(t,K,\bar{G}^{mv}(t)).\label{def_bK_MV}
\end{align}
Moreover, the optimal Lagrange multiplier $\lambda^*$ is,
\begin{align}
\lambda^*=\frac{d-x_0\bar{G}^{mv}(0)\rho(0)}{1-\bar{G}^{mv}(0)\rho(0)^2}. \label{opt_lagrange_MV}
\end{align}
\end{theorem}

\begin{proof}
Applying Theorem \ref{thm_PLQ}, we obtain the optimal policy of problem $(\overline{MV}(\lambda))$ for any fix $\lambda$,
\begin{align}
u^*(t)=
\begin{dcases}
 \hK^{mv}(t) z(t)~~\textrm{if}~~z(t) \geq 0,\\
-\bK^{mv}(t) z(t)~~\textrm{if}~~z(t) < 0.
\end{dcases}\label{mv_opt_u-1}
\end{align}
The remaining task is to identify the optimal Lagrangian multiplier $\lambda^*$. From Theorem \ref{thm_PLQ}, it not hard to derive the optimal value of problem $(\overline{MV}(\lambda))$ for any fix $\lambda$ as,
\begin{align*}
v(\overline{MV}(\lambda))=
\begin{dcases}
\hat{G}^{mv}(0)( x_0-\lambda \rho(0) )^2-(\lambda-d)^2 \\
~~~~~~~~~~\textrm{if}~~x(0)-\lambda \rho(0) \geq 0, \\
\bar{G}^{mv}(0)( x_0-\lambda \rho(0) )^2-(\lambda-d)^2 \\
~~~~~~~~~~\textrm{if}~~x(0)-\lambda \rho(0)  < 0,
\end{dcases}
\end{align*}
Then, the optimal Lagrangian multiplier $\lambda^*$ can be identified by maximizing $\lambda^*=\max_{\lambda\in\bR^n} v(\overline{MV}(\lambda))$. According to (\ref{MV_hG_convex}) and (\ref{MV_bG_convex}), we find that $v(\overline{MV}(\lambda))$ is a piece-wise concave function with respect to $\lambda$. Therefore, the optimal Lagrangian multiplier $\lambda^*$ can be derived as (\ref{opt_lagrange_MV}). The optimal portfolio decision (\ref{mv_opt_u}) is achieved by replacing $z(t)$ with $x(t)$ in (\ref{mv_opt_u-1}).
\end{proof}

Usually, the investor is interested in the MV efficient frontier, i.e., the Parato optimal set of expected terminal wealth $d$ and the associated minimum variance $\Var[x(T)]$.  Substituting $\lambda^*$ back to  $v(\widehat{MV}(\lambda))$ yields the semi-analytical expression of the MV efficient frontier as follows,
\begin{align*}
\Var[x^*&(T)]= v(\widehat{MV}(\lambda^*)-\int_{0}^{T}\E[(u^*(t))^{\prime}R(t)u^*(t)]dt \\
=& \frac{\bG^{mv}_0\rho(0)^2}{ 1 -\bG^{mv}_0\rho(0)^2}\big( \E[x^*(T)] -x_0\rho(0) \big)^2 -\int_{0}^{T}\E[(u^*(t))^{\prime}R(t)u^*(t)]dt,
\end{align*}
with $\E[x^*(T)]\geq x_0 \rho(0)^{-1}$. Note that the second term can be evaluated by the Monte Carlo simulation method after computing all $\bK^{mv}(t)$ and $\hK^{mv}(t)$.

\section{Illustrative Examples} \label{se_example}
In this section, we provide examples to illustrate the solution schemes developed in previous sections for problems ($\cP_{\LQ}^{T}$), ($\cP_{\LQ}^{\infty}$) and (MV).

\begin{example} \label{exam_LQ}
We consider an example of problem ($\cP_{\LQ}^{T}$) whose parameters are set as same as the one given in \cite{ChenZhou:2004}, i.e., the system parameters are,\footnote{We normalize the system matrices by $0.1$.}$A(t)$$=$$0.2$, $B(t)$$=$$(-5,-10,20)^{\prime}$, $C(t)$$=$$(-0.84,-3.78,0.849)^{\prime}$, and
$D(t)$$=$$\left[
  \begin{array}{ccc}
    6.85   &   11.22   &   -1.98\\
   -8.78   &   13.24   &   -5.44\\
    0.68   &   14.53   &   -2.32
  \end{array}
\right]$, for $t\in [0,T]$. The cost matrices are
\begin{align*}
R(t)=\left[
  \begin{array}{ccc}
   3.0  &   0   &   0 \\
    0   &  5.0  &   0 \\
    0   &   0   &  4.0
  \end{array}
\right],~~
S(t)=\left[
  \begin{array}{ccc}
    0.1 \\
    0.4 \\
    0.5
  \end{array}
\right]
\end{align*}
and $q(t)$ $=$ $10$ for $t$ $\in$ $[0,T)$ with $q_T$ $=$ $0$. The control horizon is $T$ $=$ $0.1$. We consider the following upper and lower bounded control constraint,
\begin{align*}
\underline{d}(t)|x(t)| \leq u(t) \leq \bar{d}(t)|x(t)|
\end{align*}
with $\underline{d}(t)$ $=$ $\left[-0.2,-0.2, -0.2\right]^{\prime}$ and $
\bar{d}(t)$ $=$ $\left[0.5, 0.5, 0.5 \right]^{\prime}$ for $t\in [0,T]$. Using the Theorem \ref{thm_PLQ}, we can obtain the optimal control for problem $(\cP_{\LQ}^{T})$ as $u^*(t)$ $=$ $\hat{K}^*(t)x(t)\1_{\{x(t)\geq 0\}}$ $-$ $\bar{K}^*(t)x(t) \1_{\{x(t)<0\}}$ where $\hat{K}^*(t)$ and $\bar{K}^*(t)$ are given in Figure \ref{exam_LQ_hbK}. Furthermore, the value function at time $t$ is $v(t,x(t))$ $=$ $\hat{G}(t)x(t)^2\1_{\{x(t)\geq 0\}}$ $+$ $\bar{G}(t)x(t)^2 \1_{\{x(t)<0\}}$, where $\hG(t)$ and $\bG(t)$ are given in Figure \ref{exam_LQ_hbG}.

We then extend the control horizon to $T=\infty$ and consider control problem ($\cP_{LQ}^{\infty}$). Based on Theorem \ref{thm_PLQ_inf} and Algorithm \ref{alg:hGbG}, we can identify $\hG^*=0.7191$ and $\bG^*=1.2032$. The associated stationary optimal control is $u^{\infty}(t)$ $=$ $\hat{K}^*x(t)\1_{\{ x(t)\geq 0 \} }$ $+$ $\bar{K}^*x(t) \1_{\{x(t)<0\} }$, for all $t\in[0,\infty)$, where $\hat{K}^*$ and $\bar{K}^*$ are given by,
\begin{align*}
\hat{K}^*=\left[
  \begin{array}{ccc}
     0.3815\\
     0.2032 \\
   - 0.2000 \\
  \end{array}
\right]
~~\textrm{and}~~
\bar{K}^*=\left[
  \begin{array}{ccc}
   -0.2000\\
   -0.2000\\
    0.1689
  \end{array}
\right].
\end{align*}
Note that whether the equations (\ref{def_equ_hG}) and (\ref{def_equ_bG}) admit the solution depends on the system parameters and the constraints. Figure \ref{exam_LQ_hFbF} plots the functions $\hat{F}(\hG)$ and $\bar{F}(\bG)$ defined in (\ref{def_hF}) and (\ref{def_bF}), respectively, for Example \ref{exam_LQ}. Based on the above parameters, we find that equations $\hat{F}(\hG)=0$ and $\bar{F}(\bG)=0$ admit the solutions, $\hG^*=0.7191$ and $\bG^*=1.2032$.(e.g., the line marked with circle) However, if we change the constraints to $\underline{d}(t)$ $=$ $\left[0.2,0.2, 0.2\right]^{\prime}$ and $
\bar{d}(t)$ $=$ $\left[0.8, 0.8, 0.8 \right]^{\prime}$ for $t\in [0,T]$, both functions  $\hat{F}(\hG)$ and $\bar{F}(\bG)$ do not admit solutions, which is plotted by `$*$' in Figure \ref{exam_LQ_hFbF}.

\begin{figure}
  \centering
  \includegraphics[width=260pt]{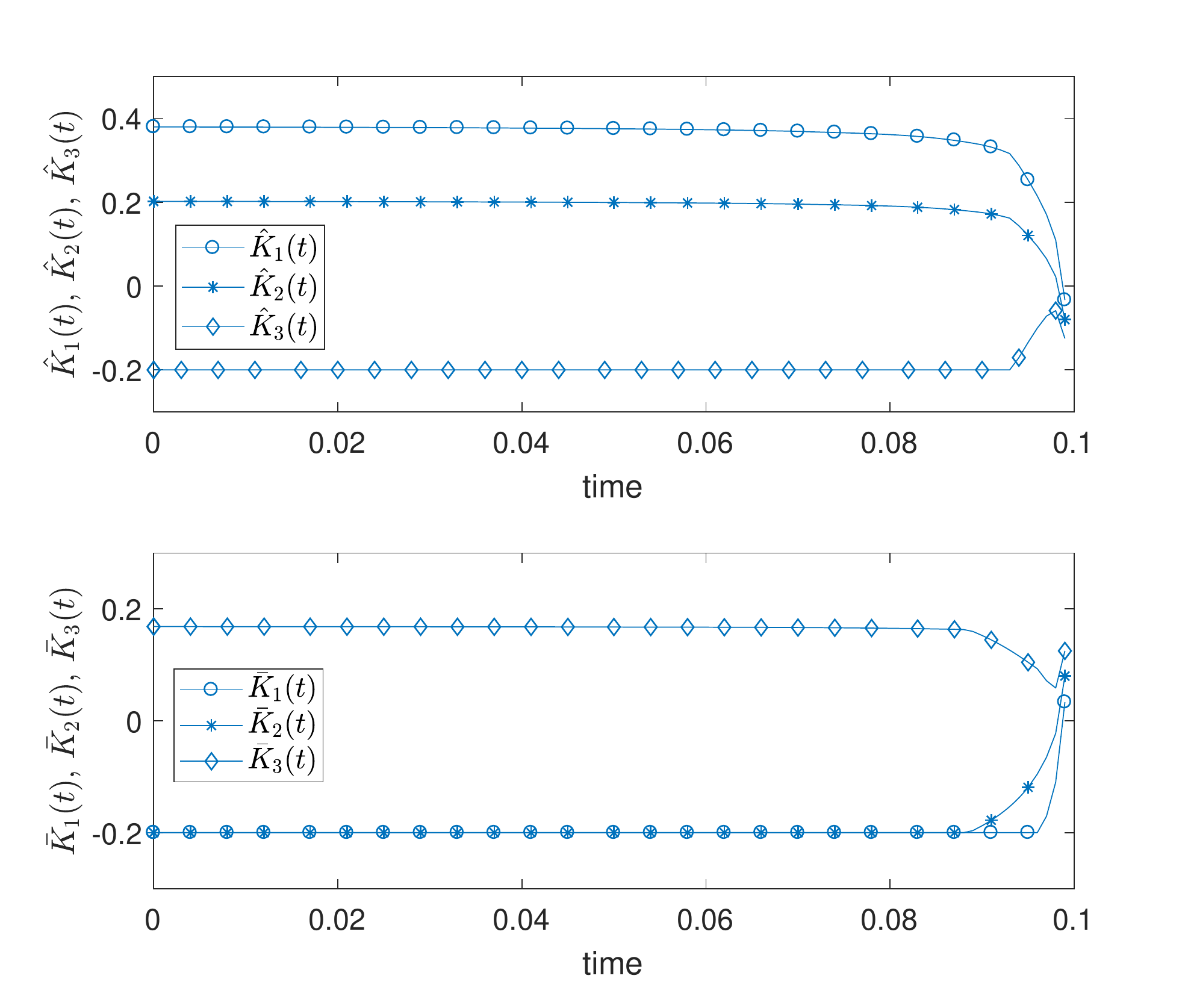}\\
  \caption{Feedback gains $\hat{K}^*(t)$ and $\bar{K}^*(t)$} \label{exam_LQ_hbK}
\end{figure}

\begin{figure}
  \centering
  \includegraphics[width=260pt]{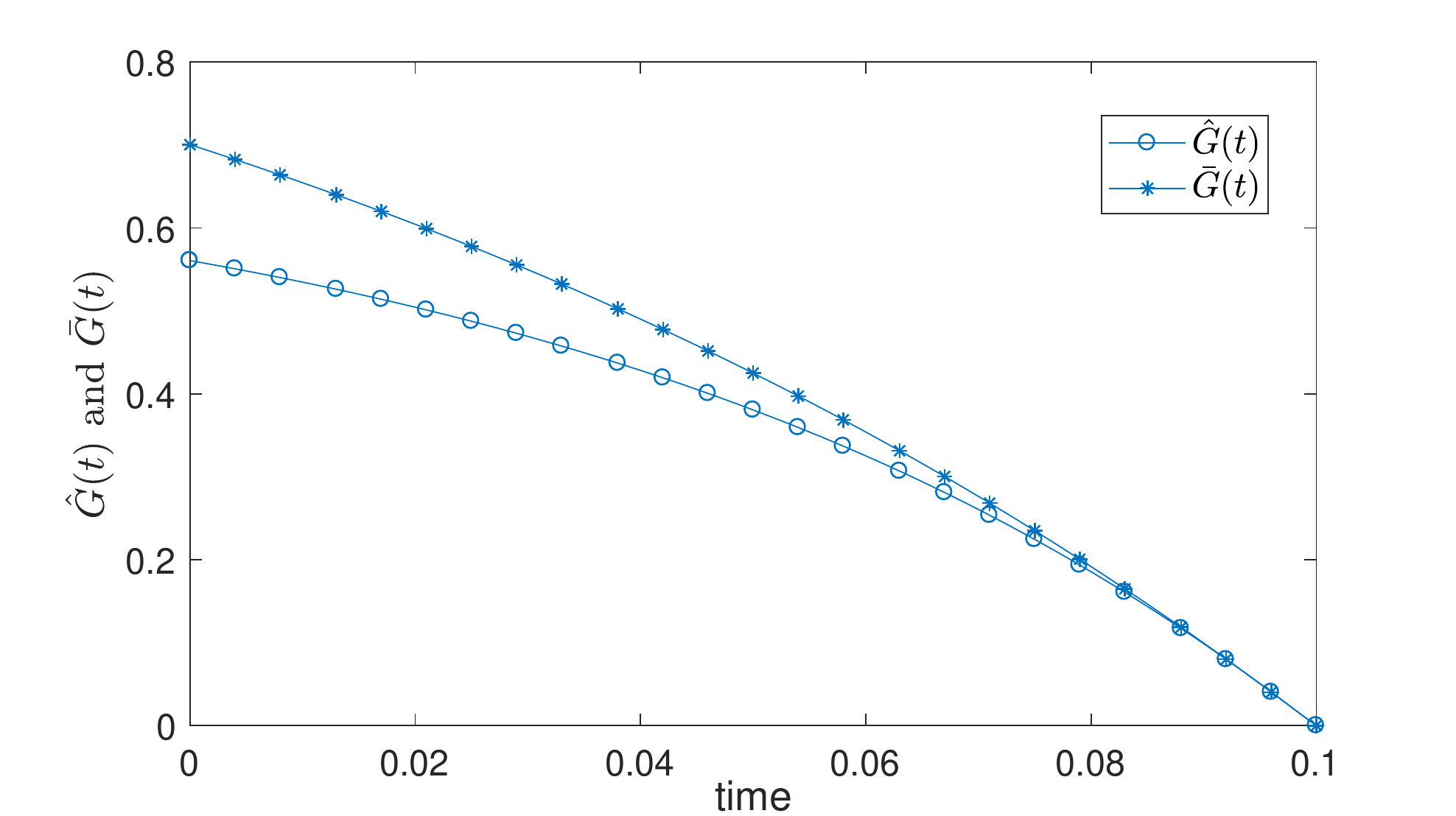}\\
  \caption{Solutions of $\hat{G}(t)$ and $\bar{G}(t)$} \label{exam_LQ_hbG}
\end{figure}

\begin{figure}
  \centering
  \includegraphics[width=260pt]{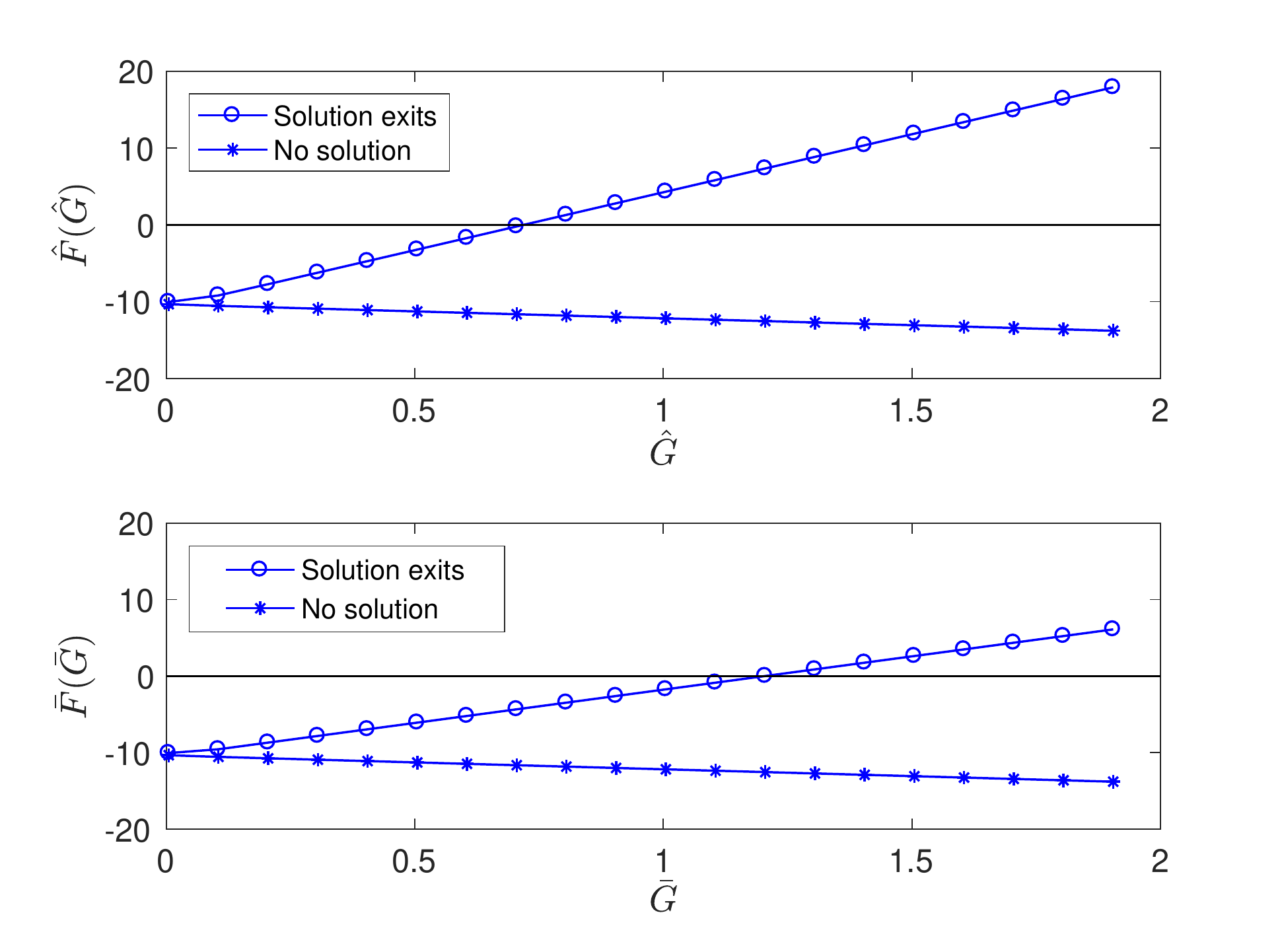}\\{}
  \caption{ Functions $\hat{F}(\hG)$ and $\bar{F}(\bG)$} \label{exam_LQ_hFbF}
\end{figure}

\end{example}

\begin{example}\label{exam_mv}
We consider an example of the MV portfolio selection model studied in Section \ref{se_application_MV}. We use the financial indices of six industrial sectors in U.S stock market as the risky assets, namely, the sectors of  Toy $\&$ Recreation, Communication, Ship Building, Coal, Gold, and Industrial Mining. \footnote{The detail instruction of these industrial indices and historical data can be found from
\url{http://mba.tuck.dartmouth.edu/pages/faculty/ken.french/Data_Library/det_48_ind_port.html}.} Based on the historical data of monthly return from January 2008 to December 2013, we can estimate the parameters of average return rate and volatility, e.g., we set
 $\mu(t)$ $=$ $\bar{\mu}$ and $\sigma(t)$ $=$ $\bar{\sigma}$ for all $t\in [0,T]$ with
\begin{align*}
\bar{\mu}&=\left(
       \begin{array}{cccccc}
         0.0321 & 0.0123 & 0.0217 & 0.0217 & 0.0282 & 0.0146
       \end{array}
     \right)~\textrm{and}\\
\bar{\sigma}&=\left(
       \begin{array}{cccccc}
       0.0845  &  0.0062  &  0.0166  &  0.0069  &  0.0087  &  0.0068 \\
       0.0062  &  0.0327  &  0.0148  &  0.0124  &  0.0067  &  0.0054 \\
       0.0166  &  0.0148  &  0.0589  &  0.0271  &  0.0204  &  0.0088 \\
       0.0069  &  0.0124  &  0.0271  &  0.1015  &  0.0215  &  0.0501 \\
       0.0087  &  0.0067  &  0.0204  &  0.0215  &  0.0609  &  0.0120 \\
       0.0068  &  0.0054  &  0.0088  &  0.0501  &  0.0120  &  0.1313
       \end{array}
     \right)。
\end{align*}
Different from the MV portfolio model studied in \cite{LiZhouLim:2002}, in which the no-shorting constraints are imposed in all risky assets, our model allows the no-shorting constraints  be only imposed on some of the assets, i.e., we set no-short constraints as, $u_2(t)\geq 0$, $u_3(t)\geq 0$, $u_4(t)\geq 0$ and $u_6(t)\geq 0$ for $t\in [0,T]$. The remained assets have no constraints, i.e., $u_1(t)$ and $u_5(t)$ are free of constraints. These constraints can be easily represented in matrix formulation $H(t) u(t)\leq \0$, i.e., by setting the diagonal elements of $H(t)$ to be $-1$ except the first and fifth elements and setting all the other elements to be $0$. We also assume that there is a risk free asset with the monthly return rate being $0.25\%$. The investor enters market with an initial wealth $x_0$ $=$ $100$ and hopes for an expected terminal wealth level $d$ $=$ $130$ with the investment horizon being $T$ $=$ $12$ months. We also include a penalty matrix on portfolio as $R(t) = 10^{-2}\I$ for all $t\in[0,T]$.

Based on Theorem \ref{thm_sol_MV}, we can compute the optimal Lagrangian multiplier as $\lambda^*=189.78$ and the optimal portfolio policy is,
\begin{align*}
u^*(t)=
\begin{cases}
\hat{K}^{mv}(t) \Big(x(t)-189.78\times e^{0.0025(t-T)}\Big)\\
~~~~~~~~~~~~~~~\textrm{if}~~x(t)-189.78\times e^{0.0025(t-T)} \geq 0,\\
-\bar{K}^{mv}(t)\Big(x(t)-189.78\times e^{0.0025(t-T)}\Big)\\
~~~~~~~~~~~~~~~\textrm{if}~~x(t)-189.78\times e^{0.0025(t-T)} < 0.\\
\end{cases}
\end{align*}
where $\hK^{mv}(t)$ and $\bK^{mv}(t)$ are plotted in Figure \ref{fig:mv_hKbK} for $t\in [0,T]$.
\begin{figure}
  \centering
  \includegraphics[width=250pt]{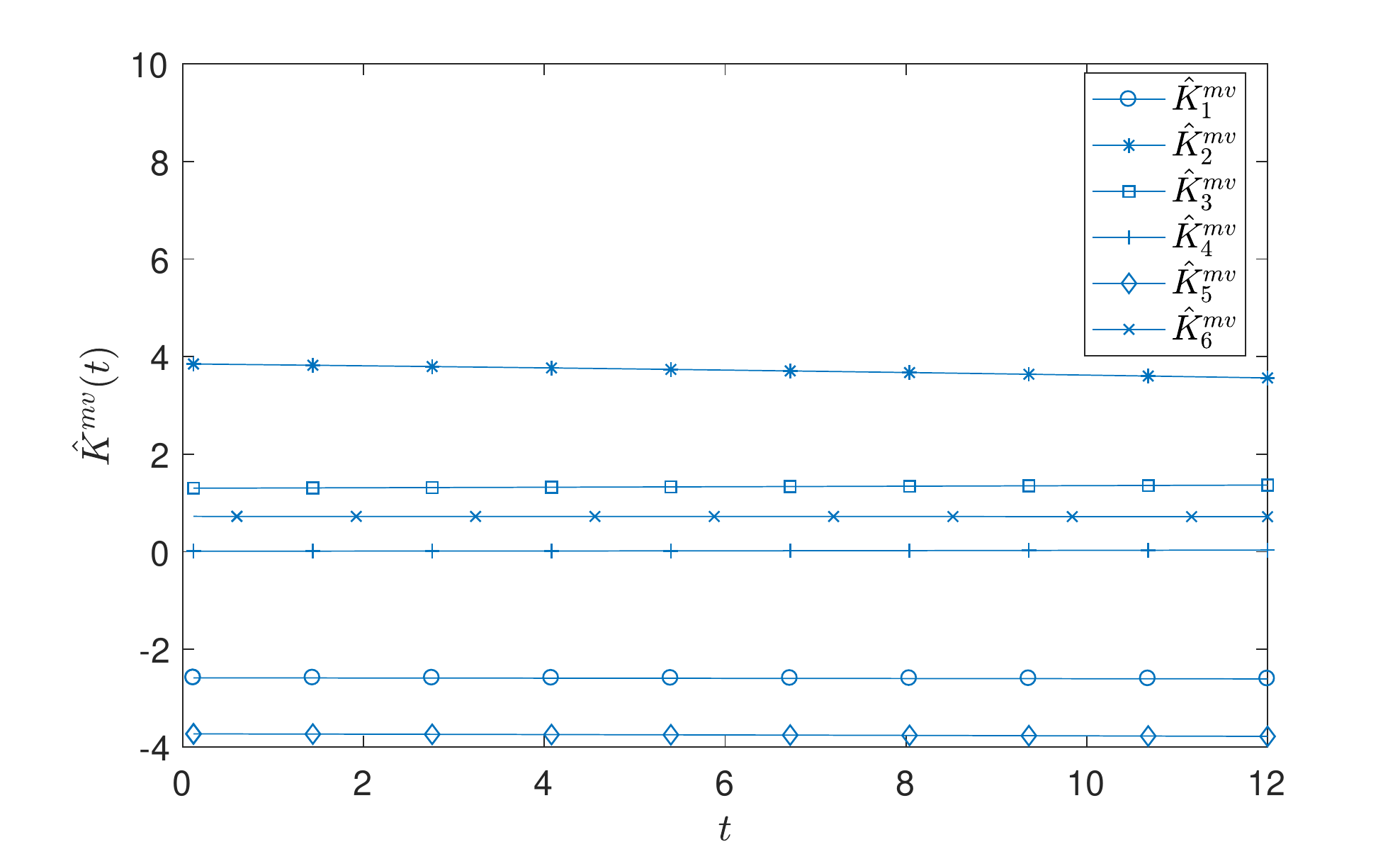}\\
  \includegraphics[width=250pt]{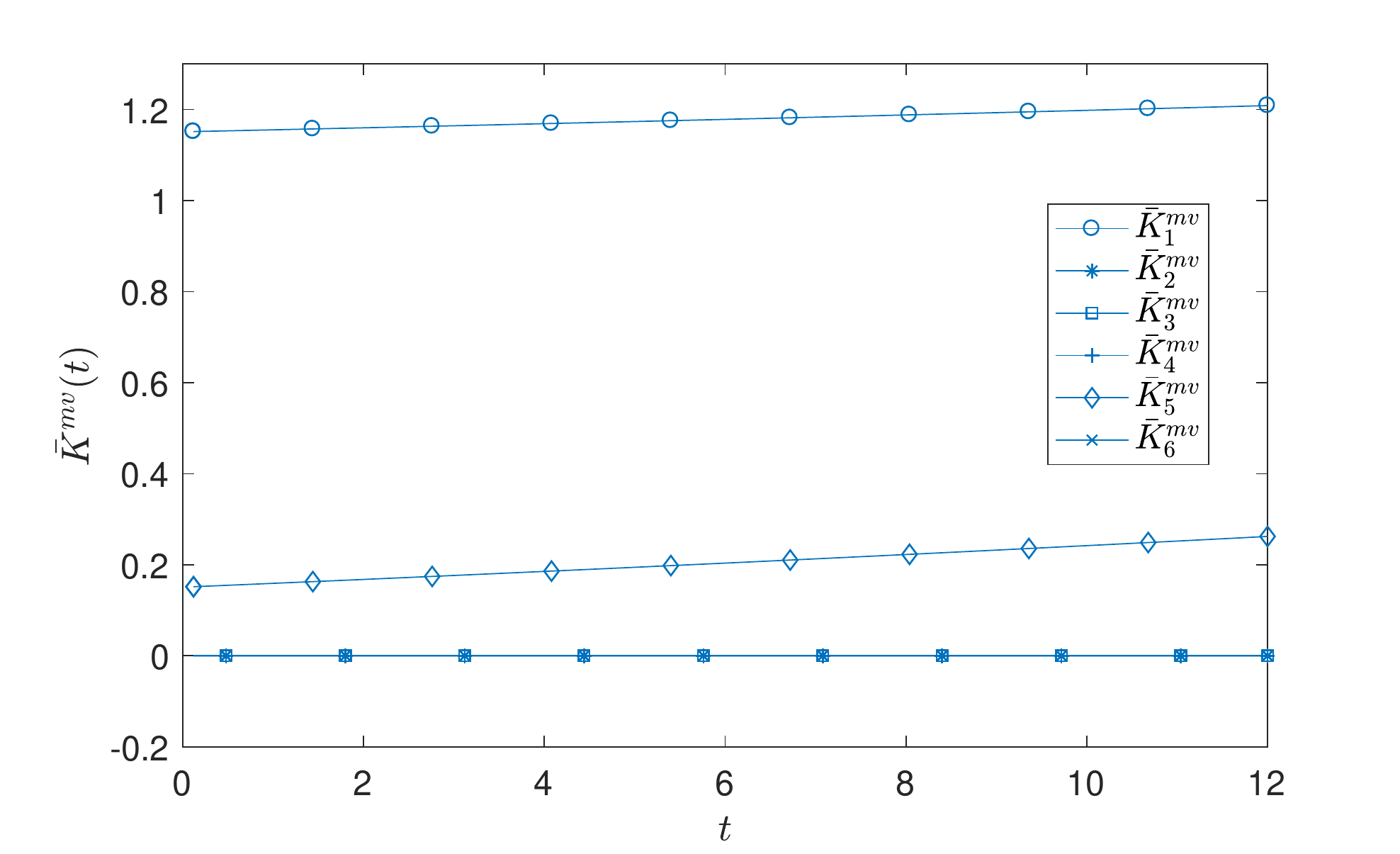}
  \caption{The outputs of $\hat{K}^{mv}(t)$ and $\bar{K}^{mv}(t)$} \label{fig:mv_hKbK}
\end{figure}

We use the buy-and-hold portfolio policy as a benchmark, i.e., we treat the total $T=12$ month as the one period static MV portfolio selection model. All the constraints and parameters are set as same as our dynamic MV portfolio selection model $(\MV)$. Solving such a problem provides us the buy-and-hold portfolio policy. Figure \ref{fig:mv_wealth} plots the realizations of the wealth processes generated by our model $(\MV)$ and the static benchmark for the identical sample path of the price process. Figure \ref{fig:mv_wealth} shows that our model performs better than the static benchmark model. To compare the performance of these two models, we plot the mean-variance efficient fronter of the terminal wealth $x(T)$ for two models in Figure \ref{fig:mv_eff}. The efficient frontier describes the Parato optimal set of the mean and standard deviation of the wealth. \footnote{In this example, the efficient frontier is plotted in the following way. We vary the expected terminal wealth level $d$ and compute the correspondent optimal portfolio. Once we have the portfolio, we can compute the standard deviation of the terminal wealth.} Figure \ref{fig:mv_eff} shows that our dynamic MV portfolio policy achieves lower level of risk than the static MV portfolio policy for the same level of expected terminal wealth.

\begin{figure}
\centering
  \includegraphics[width=250pt]{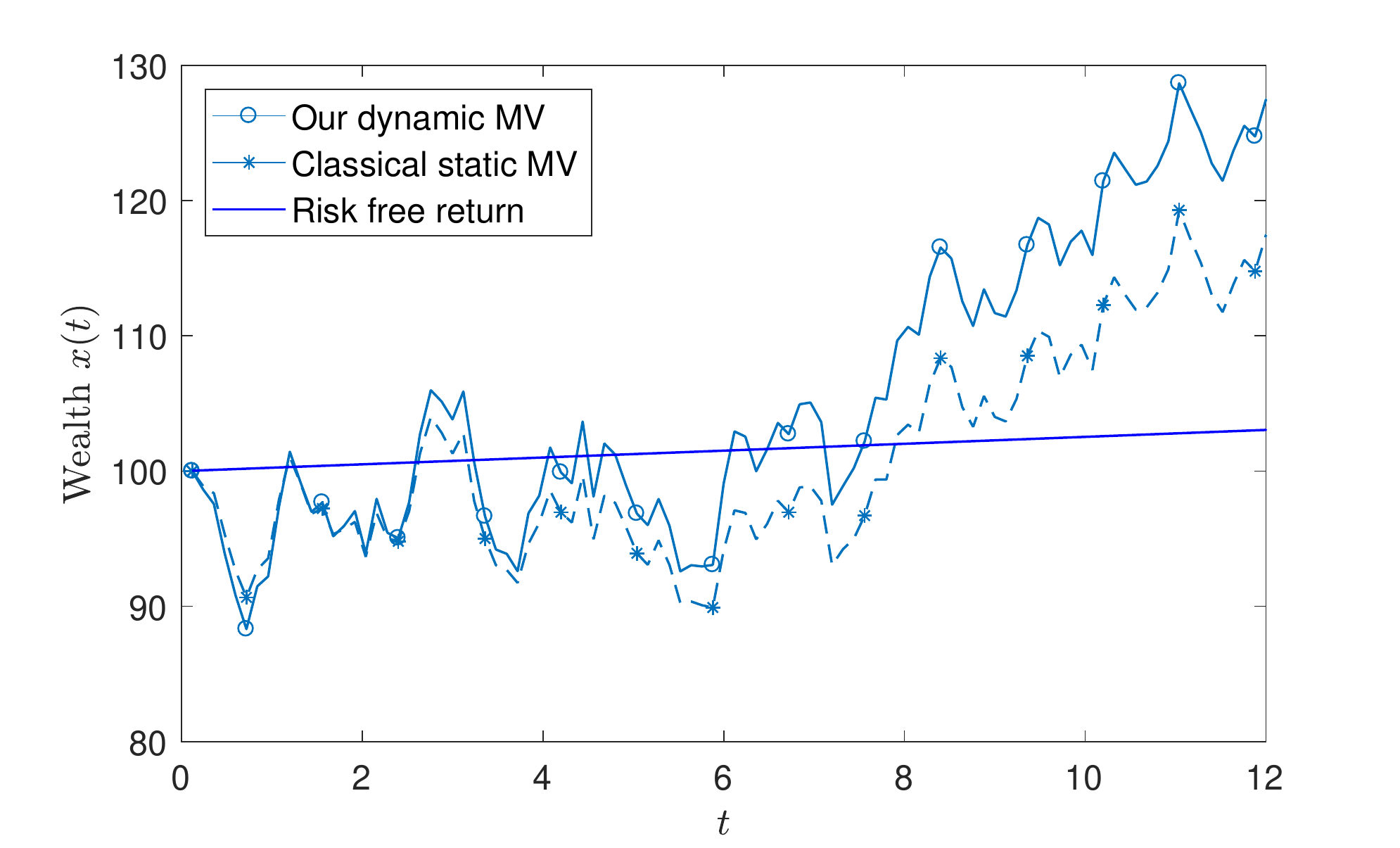}\\
    \caption{The realized wealth processes} \label{fig:mv_wealth}
\end{figure}

\begin{figure}
\centering
  \includegraphics[width=250pt]{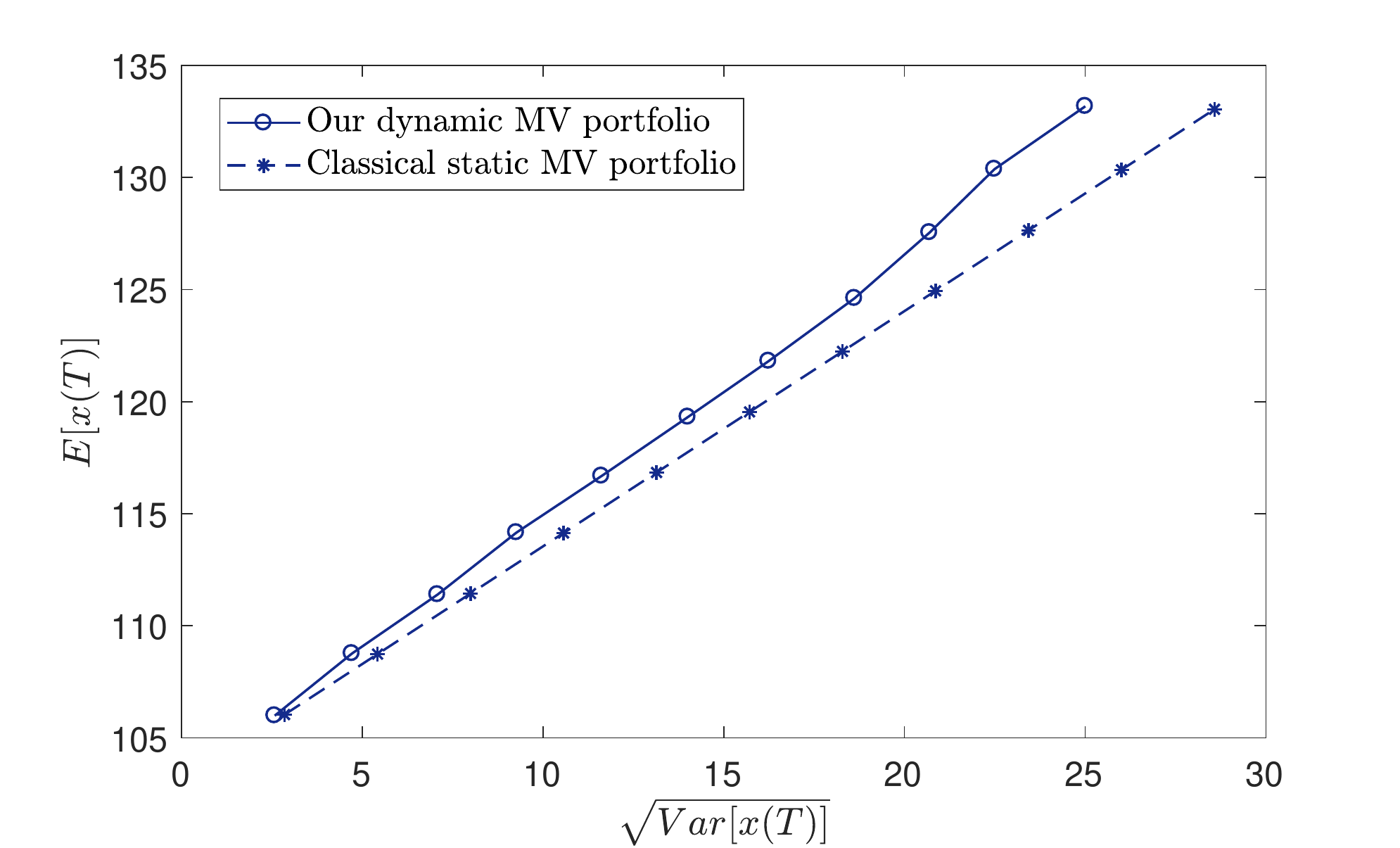}\\
    \caption{The efficient frontiers} \label{fig:mv_eff}
\end{figure}
\end{example}

\section{Conclusion} \label{se_conclusion}
In this paper, we have developed explicit solutions for the continuous-time constrained stochastic LQ control problem. Our model includes a general class of linear constraints on both control and state variables, which is able to model various control constraints, such as the positive and negative constraints, the cone constraints and the state-dependent upper and lower bound constraints. The structure property of this type of problem helps us to develop the state separation result for the control policy, which plays a key role in deriving the explicit control policy. We also extend our results to the problem with infinite control horizon. Under some conditions, the stationary optimal control policy can be achieved. These results have been applied to the dynamic constrained mean-variance portfolio selection problem. One possible future research is to extend our results to the stochastic control problem with partial moments as the objective function, which has various applications in portfolio management such as the mean-downside-risk portfolio optimization model.

\bibliographystyle{plain}        
\bibliography{cone_LQ_automatica}

\begin{thebibliography}{10}

\bibitem{AndersonLQ:2014}
B.~D.~O. Anderson and J.~B. Moore.
\newblock {\em Optimal Control: Linear Quadratic Methods}.
\newblock Dover Publications, INC. Mineola, New York.

\bibitem{Bemporad:2002}
A.~Bemporad, M.~Morari, V.~Dua, and E.~N. Pistikopoulos.
\newblock The explicit linear quadratic regulator for constrained systems.
\newblock {\em Automatica}, 38(1):3--20, 2002.

\bibitem{BernardiniBemporad:2012}
D.~Bernardini and A.~Bemporad.
\newblock Stabilizing model predictive control of stochastic constrained linear
  systems.
\newblock {\em IEEE Trans. Autom. Control}, 57(6):1468--1480, 2012.

\bibitem{Bismut:1976}
J.~M. Bismut.
\newblock Linear quadratic optimal stochastic control with random coefficients.
\newblock {\em SIAM J. Control Optim.}, 14(3):419--444, 1976.

\bibitem{Boyd:2004}
S.~Boyd and L.~Vandenberghe.
\newblock {\em Convex optimization}.
\newblock Cambridge University Press, 2004.

\bibitem{Campbell:1982}
S.~L. Campbell.
\newblock On positive controllers and linear quadratic optimal control
  problems.
\newblock {\em Internat. J. Control}, 36(5):885--888, 1982.

\bibitem{ChenZhou:2004}
X.~Chen and X.~Y. Zhou.
\newblock Stochastic linear-quadratic control with conic control constraints on
  an infinite time horizon.
\newblock {\em SIAM J. Control Optim.}, 43(3):1120--1150, 2004.

\bibitem{CuiGaoLiLi:2014}
X.~Y. Cui, J.~J. Gao, X.~Li, and D.~Li.
\newblock Optimal multi-period mean-variance policy under no-shorting
  constraint.
\newblock {\em Euro. J. Oper. Res.}, 234(2):459--468, 2014.

\bibitem{CuiLiLi:2017}
X.~Y. Cui, D.~Li, and X.~Li.
\newblock Mean-variance policy for discrete-time cone constrained markets: The
  consistency in efficiency and minimum-variance signed supermartingale
  measure.
\newblock {\em Math. Finance}, 27(2):471--504, 2017.

\bibitem{GaoLi:2011}
J.~J. Gao and D.~Li.
\newblock Cardinality constrained linear-quadratic optimal control.
\newblock {\em IEEE Trans. Autom. Control}, 56(8):1936--1941, 2011.

\bibitem{HeemelsEijndhovenStoorvogel:1998}
W.~P. Heemels, S.~V. Eijndhoven, and A.~A. Stoorvogel.
\newblock Linear quadratic regulator problem with positive controls.
\newblock {\em Internat. J. Control}, 70(4):551--578, 1998.

\bibitem{HuHuangNie:2017}
Y.~Hu, J.~H. Huang, and T.~Y. Nie.
\newblock Linear-quadratic-gaussian mixed mean-field games with heterogeneous
  input constraints.
\newblock {\em submitted for pubplication}, 2017.

\bibitem{HuZhou:2005}
Y.~Hu and X.~Y. Zhou.
\newblock Constrained stochastic lq control with random coefficients, and
  application to portfolio selection.
\newblock {\em SIAM J. Control Optim.}, 44(2):444--466, 2005.

\bibitem{Kalman:1960}
R.~E. Kalman.
\newblock Contribution to the theory of optimal control.
\newblock {\em Bol. Soc. Mat. Mexicana}, 5(63):102--119, 1960.

\bibitem{LiNg:2000}
D.~Li and W.~L. Ng.
\newblock Optimial dynamic portfolio selection: multiperiod mean-variance
  formulation.
\newblock {\em Mathematical Finance}, 10:387--406, 2000.

\bibitem{LiZhouLim:2002}
X.~Li, X.~Y. Zhou, and A.~E.~B. Lim.
\newblock Dynamic mean-variance portfolio selection with no shorting
  constraints.
\newblock {\em SIAM J. Control Optim.}, 40(5):1540--1555, 2002.

\bibitem{Mesbah:2016}
A.~Mesbah.
\newblock Stochastic model predictive control: An overview and perspectives for
  future research.
\newblock {\em IEEE Control Systems}, 36(6):30--44, 2016.

\bibitem{Patrinos:2014}
P.~Patrinos, P.~Sopasakis, H.~Sarimveis, and A.~Bemporad.
\newblock Stochastic model predictive control fro constrained discrete-time
  markovian switching systems.
\newblock {\em Automatica}, 50:2504--2514, 2014.

\bibitem{Pham:2009}
H.~Pham.
\newblock {\em Continuous-time Stochastic Control and Optimization with
  Financial Applications}.
\newblock Springer, 2009.

\bibitem{Primb:2009}
J.~A. Primbs and C.~H. Sung.
\newblock Stochastic receding horizon control of contrained linear systems with
  state and control multiplicative noise.
\newblock {\em IEEE Trans. Automat. Contr.}, 54(2):221--230, 2009.

\bibitem{Wonham:1969}
W.~M. Wonham.
\newblock On a matrix riccati equation of stochastic control.
\newblock {\em SIAM J. Control}, 6(4):681--697, 1969.

\bibitem{WuGaoLiShi:2017}
W.~P. Wu, J.~J. Gao, D.~Li, and Y.~Shi.
\newblock Explicit solution for constrained scalar-state stochastic
  linear-quadratic control with multiplicative noise.
\newblock {\em summitted for publication}, 2017.

\bibitem{YongZhou:1999}
J.~M. Yong and X.~Y. Zhou.
\newblock {\em Stochastic controls: Hamiltonian systems and HJB equations}.
\newblock Springer, 1999.

\bibitem{ZhouLi:2000}
X.~Y. Zhou and D.~Li.
\newblock Continuous-time mean-variance portfolio selection: A stochastic lq
  framework.
\newblock {\em Appl Math Optim}, 42(1):19--33, 2000.

\bibitem{ZhouYongLi:1997}
X.~Y. Zhou, J.~M. Yong, and X.~J. Li.
\newblock Stochastic verification theorems within the framework of viscosity
  solutions.
\newblock {\em SIAM J. Control Optim.}, 35(1):243--253, 1997.

\end{thebibliography}

\end{document}